\pdfoutput=1

\documentclass[12pt]{article}

\usepackage[square,sort,comma,numbers]{natbib} 
\usepackage{hyperref}
\hypersetup{colorlinks,citecolor=blue,linkcolor=blue}

\usepackage{enumerate,tikz}										
\usepackage{ragged2e}

\usepackage[labelfont=bf]{caption}

\usepackage{amsmath,amsfonts,amssymb,amsthm,booktabs,color,epsfig,graphicx,url}

\theoremstyle{plain}
\newtheorem{thm}{Theorem}

\newtheorem{prop}{Proposition}
\newtheorem{lemma}{Lemma}
\newtheorem{cor}{Corollary}

\theoremstyle{definition}

\newtheorem{exa}{Example}

\newcommand{\aaa}{\mathbf{a}}

\newcommand{\dd}{\mathbf{d}}

\newcommand{\uu}{\mathbf{u}}
\newcommand{\vv}{\mathbf{v}}
\newcommand{\w}{\mathbf{w}}
\newcommand{\x}{\mathbf{x}}
\newcommand{\y}{\mathbf{y}}
\newcommand{\zz}{\mathbf{z}}
\newcommand{\bta}{\boldsymbol{\beta}}

\newcommand{\eeps}{\boldsymbol{\epsilon}}
\newcommand{\veeps}{\boldsymbol{\varepsilon}}

\newcommand{\tht}{\boldsymbol{\theta}}

\newcommand{\A}{\mathbf{A}}
\newcommand{\B}{\mathbf{B}}
\newcommand{\C}{\mathbf{C}}

\newcommand{\HH}{\mathbf{H}}
\newcommand{\I}{\mathbf{I}}
\newcommand{\J}{\mathbf{J}}
\newcommand{\KK}{\mathbf{K}}
\newcommand{\PP}{\mathbf{P}}
\newcommand{\Q}{\mathbf{Q}}

\newcommand{\V}{\mathbf{V}}

\newcommand{\X}{\mathbf{X}}

\newcommand{\Z}{\mathbf{Z}}

\newcommand{\LAM}{\boldsymbol{\Lambda}}
\newcommand{\OMEGA}{\boldsymbol{\Omega}}
\newcommand{\PSI}{\boldsymbol{\Psi}}

\DeclareMathOperator*{\argmax}{argmax}
\DeclareMathOperator*{\argmin}{argmin}

\newcommand{\tr}{\text{\rm tr}}
\newcommand{\rk}{{\text{\rm rk}}}

\newcommand{\Prob}{{{\rm Pr}}}
\newcommand{\ThetaC}{\Theta}

\makeatletter \let\@fnsymbol\@arabic \makeatother 					

\allowdisplaybreaks
\title{Uniformly valid inference based on the Lasso \\in linear mixed
models} 
\date{} 
\author{Peter Kramlinger\thanks{Department of Statistics, University of California Davis, One Shields Ave, Davis, CA 95616, United States}
\and Ulrike Schneider\thanks{ulrike.schneider@tuwien.ac.at, Institute
of Statistics and Mathematical Methods in Economics, Technische
Universit\"at Wien, Wiedner Hauptstr. 8, 1040 Wien, Austria}
\and Tatyana Krivobokova\thanks{tatyana.krivobokova@univie.ac.at,
Department of Statistics and Operations Research, Universit\"at Wien,
Oskar-Morgenstern-Platz 1, 1090 Wien, Austria} }

\begin{document}

\maketitle
\setcounter{footnote}{3}

\begin{abstract}
Linear mixed models (LMMs) are suitable for clustered data and are
common in biometrics, medicine, survey statistics, and many other
fields. In those applications, it is essential to carry out valid
inference after selecting a subset of the available variables. We
construct confidence sets for the fixed effects in Gaussian LMMs that
are based on Lasso-type estimators. Aside from providing confidence
regions, this also allows for quantification of the joint uncertainty
of both variable selection and parameter estimation in the procedure.
To show that the resulting confidence sets for the fixed effects are
uniformly valid over the parameter spaces of both the regression
coefficients and the covariance parameters, we also prove the novel
result on uniform Cram\'er consistency of the restricted maximum
likelihood (REML) estimators of the covariance parameters. The
superiority of the constructed confidence sets to na\"ive
post-selection procedures is validated in simulations and illustrated
with a study of the acid-neutralization capacity of lakes in the
United States.
\bigskip


\end{abstract}

\section{Introduction} 

Linear mixed models (LMMs) are regression models that incorporate
dependency structures that occur in clustered data. They are widely
applied in many empirical sciences ranging from genetics
\citep{Henderson1950} to survey statistics \citep{Pfefferman2013} and
many more. Comprehensive reviews can be found in \citep{Demidenko2004, Pinheiro2000}. The standard LMM can be written as
\begin{equation}
\begin{gathered}\label{LMM}
\y_i = \X_i\bta_0 + \Z_i \vv_i + \veeps_i,\;\;i\in\{1,\ldots,m\}, \quad
\veeps_i \sim
\mathcal{N}_{n_i}\{\boldsymbol{0}_{n_i},\OMEGA_i(\tht_0)\}, \quad
\vv_i \sim \mathcal{N}_{q}\{\boldsymbol{0}_{q},\PSI(\tht_0)\},
\end{gathered}
\end{equation}
with observations $\y_i\in\mathbb{R}^{n_i}$, known fixed covariates
$\X_i\in\mathbb{R}^{n_i\times p}$, $\Z_i\in\mathbb{R}^{n_i\times q}$,
and $\vv_i\in\mathbb{R}^{q}$ and $\veeps_i\in\mathbb{R}^{n_i}$ being
independent random variables and each independently
distributed for all clusters $i\in\{1,\dots, m\}$. 
The regression coefficient vector $\bta_0$ is referred to as fixed effects,
whereas $\vv_i$ are the random effects. While all observations share
the same regression coefficients in the fixed effects, only $n_i$
observations in each cluster share the same realization of $\vv_i$ in
the random effects. The latter can be associated, e.g., with subject-specific effects
in medical trials or group effects in longitudinal studies. The
regression coefficients $\bta_0\in\mathbb{R}^{p}$ and covariance parameters
$\tht_0 \in \ThetaC \subset \mathbb{R}_{>0}^r$ are unknown and have to
be estimated. Consistency of parameter estimators for $\bta_0$ and
$\tht_0$ is ensured if the number of clusters $m$ tends to
infinity, whereas the sample size per cluster $n_i$ may stay bounded
or grow as well \citep[Section 3.6.2]{Demidenko2004}.

In practice, often both model selection and estimation need to be
performed. Na\"ive two-stage methods for inference in regression
models consist of a variable selection step, which is based, e.g., on
an information criterion, and subsequent estimation of the parameters
and inference in the obtained model in a second step. However, this
na\"ive inferential theory ignores the first model selection step and
does not account for the additional uncertainty introduced by the
selection procedure. This issue has been laid out in detail in
\citep{Poetscher1991}. In recent years, different approaches to this
problem have emerged in the literature. The PoSi or post-selection
inference framework of \citep{Berk2013} provides a general procedure to
construct confidence regions for a best-approximation parameter that
depends on the selected model. In that regard, a similar perspective
was taken by \citep{Lee2016} who also aim to cover this surrogate
quantity rather than a `true' parameter and introduce the idea of
selective inference with a focus on the Lasso. This selective
inference approach, where the coverage holds conditionally on the
selected model, is followed in \citep{Ruegamer2022} specifically for
LMMs. Following a similar paradigm in the sense of conditional
inference, \citep{Charkhi2018} suggest an approach based on the Akaike
Information Criterion (AIC) which is carried out for LMMs in
\citep{Claeskens2021}. Contrasting this, a more classical perspective
is taken in the general debiasing approach of \citep{VdGeerEtAl14} and
related articles, e.g., \citep{Li2021} for LMMs, where the target is (a
sub-vector of) the true parameter and the coverage probability holds
unconditionally. We follow this classical approach in that we also
consider unconditional inference with the aim to cover the true
parameter. All the methods listed above are motivated from a
high-dimensional perspective. In our work, we focus on a
low-dimensional setting, that is, when the number of variables is
smaller than the number of observations.

In this work, we employ the least absolute shrinkage and selection
operator (Lasso) \citep{Tibshirani1996}, which offers a single step
approach to simultaneous variable selection and estimation for the
regression coefficients. Although in the context of the linear mixed
models, the Lasso can be applied to the selection of both fixed and
random effects \citep[]{Ibrahim2011,Bondell2010,Peng2012,Mueller2013},
we consider selection and inference for $\bta_0$ only. Selection of
covariates $\bta_0$ is fundamental to interpreting the association
between the set of covariates and the response, which is of high
relevance in most practical problems.


The usefulness and wide application of Lasso sparked interest in how
to construct confidence intervals based on the resulting Lasso
estimators. Not too surprisingly, the model selection step inherent to
the Lasso estimator leads to challenges with respect to inference. For
the Lasso and related estimators, this manifests in the general
difficulty that its finite sample and, in an appropriate framework,
also the asymptotic distribution depends on the unknown coefficient
parameters $\bta_0$ in an intricate manner \citep{PoetscherLeeb2009}.
This is problematic since honest confidence sets in the sense of
\citep{Li1989} have to attain nominal coverage over the whole parameter
space, which has been investigated in detail for the Lasso in the case of
orthogonal covariates in \citep{PoetscherSchneider10}. For a general
low-dimensional framework, where the number of parameters is smaller
than the sample size, \citep{Ewald2018} have obtained exact and
uniformly valid confidence sets. All these references consider ordinary
linear models and do not cover the case of linear mixed models.

The contribution of the present article to the LMM framework is
two-fold. First, we show that the covariance parameter $\tht_0$ can be
estimated Cram\'er-consistently in a {uniform} manner. Second, we
construct confidence regions for $\bta_0$ based on the Lasso which are
uniformly valid over the space of coefficient and covariance
parameters. Unlike an ordinary linear model, an LMM requires the
estimation of covariance parameters $\tht_0$ in addition to the
coefficient parameters $\bta_0$. The key idea is to estimate the
former with restricted maximum likelihood (REML) as this procedure
allows to separate the estimation of $\tht_0$ and $\bta_0$. A novel
uniform consistency result for the REML estimators is established,
extending the previous result of \citep{Weiss1971}. This new result is
crucial to non-trivially extend the reasoning of \citep{Ewald2018} to
construct confidence sets for $\bta_0$ that hold the coverage not
only uniformly over the space of coefficient parameters but also over
the space of covariance parameters. Thus, researchers are provided
with a method that allows valid inference based on the Lasso in LMMs.
Our analysis is provided within a low-dimensional framework where $p
< n$.

The remainder of the article is structured as follows. First, we
specify the setting and regularity conditions and state the estimation
procedure for the fixed effects $\bta_0$ in
Section~\ref{sec:setting_and_ass}. Next, in Section~\ref{sec:reml},
the estimation of covariance parameters $\tht_0$ is discussed and its
uniform consistency is established. In Section~\ref{sec:main_results},
the confidence regions for $\bta_0$ based on the Lasso are presented.
Their usefulness and limitations, and, in particular, their
superiority to na\"ive approaches is
demonstrated in a simulation study in Section~\ref{sec:sims}. A real
data example is provided in Section~\ref{sec:app}. The article
concludes with a discussion in Section~\ref{sec:disc}.

Throughout the remainder of the article, we use the following notation.
For a matrix $\A$, $\tr(\A)$ and $\rk(\A)$ denote the trace and rank
of $\A$, respectively. By $\eta_i(\A)$, we refer to the $i$-th
eigenvalue of $\A$, sorted in descending order. The symbol
$\|\mathbf{A}\|$ denotes the Frobenius norm, i.e., $\|\mathbf{A}\|^2 =
\tr(\mathbf{A}^\top\mathbf{A})$. Finally, for $\A_i \in
\mathbb{R}^{n_i\times m_i}$, $\mbox{blockdiag}(\A_i)_{i=1}^n$ denotes
the $({{\sum_{i=1}^n n_i}\times{\sum_{i=1}^n m_i}})$-dimensional
block-diagonal matrix with $\A_i$ along the diagonal block and zeroes
elsewhere.


\section{Setting and regularity conditions}
\label{sec:setting_and_ass} %

Consider a linear model with $n\in\mathbb{N}$ dependent observations
and model equation
\begin{equation} \label{GenLM}
\y = \X \bta_0 + \eeps, \quad
\eeps \sim \mathcal{N}_n\big\{\boldsymbol{0}_n, \V(\tht_0) \big\}.
\end{equation}
The covariance matrix $\V(\cdot) \in \mathbb{R}^{n \times n}$ models
the dependency amongst the observations. The vectors $\bta_0 \in
\mathbb{R}^p$ as well as $\tht_0 \in \ThetaC \subset
\mathbb{R}_{>0}^r$ are unknown and remain to be estimated. Model
(\ref{LMM}) can be represented as (\ref{GenLM}) with $n = \sum_{i=1}^m
n_i$, $\y = (\y_{n_1},\ldots,\y_{n_m})^\top$, $\X = (\X_1^\top,\ldots,
\X_m^\top)^\top$ and $\V(\tht_0) = \mbox{blockdiag}\{\Z_i \PSI(\tht_0)
\Z_i^\top + \OMEGA_i(\tht_0)\}_{i=1}^m$. 
For given tuning parameters $\lambda_1,\dots,\lambda_p$ and a REML
estimator $\hat\tht$ for $\tht_0$ (defined by (\ref{REML}) in
Section~\ref{sec:reml}), consider the Lasso estimator
\begin{equation} \label{Lasso}
\hat\bta_L = \argmin_{\bta\in\mathbb{R}^p} 
\left\{ \left\|\V(\hat\tht)^{-1/2}\left(\y - \X\bta\right)\right\|^2
+ 2\sum_{j=1}^p \lambda_j\left|\beta_j\right|\right\}
\end{equation}
for $\bta_0$. 
If $\hat\tht$ in $(\ref{Lasso})$ were replaced by
the true parameter $\tht_0$, the objective function would reduce to
the standard $\ell_1$-penalization of the Lasso.
Note that $\hat\bta_L$ simultaneously performs estimation and variable selection. %

This paper aims to provide confidence regions based on
$\hat\bta_L$ for the fixed effects $\bta_0$, which hold their
coverage probability uniformly over the space of coefficient and
covariance parameters. Formally, we find a set $M \subset
\mathbb{R}^p$ such that $\inf_{\bta_0, \tht_0}\Prob_{\bta_0, \tht_0}\{
\sqrt{n}(\textstyle\hat\bta_L-\bta_0) \in M \} \approx 1-\alpha$
for some nominal level $\alpha\in(0,1)$, which is attained uniformly
over $\bta_0\in\mathbb{R}^p$ and $\tht_0\in\ThetaC$. %
We borrow and adapt the notation of \citep{Ewald2018}. The key idea of
that article is the following: Instead of treating $\hat\bta_L$
with its intractable distribution directly, one considers limiting
versions of the objective function in $(\ref{Lasso})$, in which each
component of $\bta_0$ tends to infinity in absolute value. This
process gives $2^p$ minimizers. One can choose $M$ such that it
contains each of these minimizers with a probability of at least $1 -
\alpha$. It can be shown that this procedure does yield exact
and uniformly valid confidence sets for ordinary linear models. For
the model in (\ref{GenLM}) and optimization problem in (\ref{Lasso}),
this gives the $2^p$ minimizers
\begin{equation}\label{objective_function_limiting}
\hat\uu_\dd =
	\argmin_{\uu\in\mathbb{R}^p} \;
		 \uu^\top\hat\C\uu
		- 2\uu^\top \hat\w
		+ 2\uu^\top \LAM \dd
	= \hat\C^{-1}\Big(\hat\w - \LAM\dd\Big),
\end{equation}
where $\hat\C = n^{-1}\X^\top\V(\hat\tht)^{-1}\X$, $\hat\w =
n^{-1/2}\X^\top\V(\hat\tht)^{-1}(\y - \X\bta_0)$, $\LAM =
n^{-1/2}\text{diag}(\lambda_1, \dots, \lambda_p)$ and $\dd \in
\{-1,1\}^p$. The latter plays the role of representing the signs of the
coefficients. 
The contribution of the current paper is to prove that a similar idea
can be used for LMMs, for which the main challenge is to show that the
uniform coverage carries over to the covariance parameters and that
the estimation of these parameters does not interfere with the
required uniformity over the coefficient parameters.

Subsequently, we will denote by $\uu_\dd$, $\C$ and $\w$ the
counterparts of $\hat\uu_\dd $, $\hat\C$ and $\hat\w$,
with $\hat\tht$ replaced by $\tht_0$. As the distribution of
$\hat\uu_\dd$ is not analytically available, the proofs use that
$\uu_\dd\sim \mathcal{N}_p(-\C^{-1}\LAM\dd,\C^{-1})$. We only consider
confidence sets of ellipsoidal shape, that is, $M$ will be set to
$E(\hat\C,t) = \{\zz\in\mathbb{R}^p \;|\; \zz^\top\hat\C\zz \leq
t  \}$, $t>0$. Their specific form is used in Lemma~\ref{uniflimpiv}
to establish the minimal coverage probability over $\tht_0$. %

The following regularity conditions on the model \eqref{GenLM} are imposed.
\begin{enumerate}[(A)]

\item For a constant $c \in (1,\infty)$, $\tht_0\in\ThetaC =
\left\{\tht\in\mathbb{R}^r_{>0}\;|\; \max(\tht)/\min(\tht) \leq
c\right\}$. \label{c_theta}

\item $\V(\tht_0) = \sum_{k=1}^r \theta_{0,k} \HH_k$ with 
$\HH_k$, $k\in\{1,\dots, r-1\}$, positive semi-definite and $\HH_r$ positive
definite. 
Further, $m = \min_{k\in\{ 1, \dots, r\}}\rk(\HH_k)\rightarrow\infty$.
\label{c_linear}

\item $\rk(\X) = p<n$ and for all $\tht_0\in\Theta$: $\C = n^{-1}\X^\top\V(\tht_0)^{-1}\X\rightarrow
\tilde\C$ pointwise for finite and positive definite $\tilde
\C$ as $n \to \infty$.
\label{c_x}

\item \label{c_h} For constants
$0<\underline{\omega}\leq\overline{\omega}<\infty$ it holds that 
$\underline\omega\leq\tr\{ (\sum_{l=1}^r
\HH_l)^{-1}\HH_k\}/\rk(\HH_k)\leq\overline\omega$ 
and
$\sum_{j=1}^p \eta_j\{ (\sum_{l=1}^r\HH_l)^{-1}\HH_k\} = O(1)$, 
$k \in\{1,\ldots,r\}$.
\item \label{c_pd} For $\PP(\tht, a) =
\V(\tht)^{-1}-a\V(\tht)^{-1}\X\big\{\X^\top\V(\tht)^{-1}\X\}^{-1}\X^\top\V(\tht)^{-1} \in \mathbb{R}^{n \times n}$, $a\geq1$ and $\KK \in \mathbb{R}^{r \times r}$ with entries $(\KK)_{ij} = \tr\{\PP(\mathbf{1}_r, c)\HH_i\PP(\mathbf{1}_r, c)\HH_j\}/\sqrt{\rk(\HH_i)\rk(\HH_j)}$ it holds that $\eta_r(\KK)\geq b$, for $b>0$ constant.
\end{enumerate}

First, (\ref{c_theta}) imposes restrictions on $\tht_0$. It
encompasses the conditions of positive covariance parameters and
non-degenerativity as introduced by \citep{Jiang1996}. The condition
$\max(\tht)/\min(\tht) \leq c$ allows for all components of $\tht_0$
to be either arbitrarily small or large, but not both.
Condition~(\ref{c_linear})  ensures that the covariance matrix is not
singular and that $\tht_0$ can be estimated consistently. Next,
(\ref{c_x}) corresponds to \cite[Eq. (3.47)]{Demidenko2004} and ensures that $\bta_0$ can be estimated consistently, 
which is the key advantage in the low-dimensional setting, that is for
$p<n$. An extension to the case $n\le p$ is discussed in
Section~\ref{sec:disc}. Furthermore, (\ref{c_h}) ensures that the
normalizing sequence in Theorem~\ref{unifconv} below bounds the REML
estimator for $\tht_0$ uniformly. Finally, (\ref{c_pd}) relates to the
estimator of the covariance parameters $\tht_0$. It encompasses the
condition that $\HH_k\in\mathbb{R}^{n\times n}$, $k\in\{1,\dots, r\}$, are
linearly independent, which corresponds to the condition of
identifiability of covariance parameters from \citep{Jiang1996}. Since
all involved matrices are known, the last two conditions are typically
easy to verify, as the next two examples show. \\%

\begin{exa}
Consider a linear mixed model (\ref{LMM}) with $\X_i=\Z_i =
\mathbf{1}_{n_i}$, $\PSI(\tht_0) = \sigma_v^2$ and $\OMEGA_i(\tht_0) =
\sigma_e^2\mathbf{I}_{n_i}$, so that
$\V(\tht_0)=\sigma_v^2\HH_1+\sigma_e^2\HH_2$, with
$\HH_1=\text{blockdiag}(
\mathbf{1}_{n_1}\mathbf{1}_{n_1}^\top,\ldots,\mathbf{1}_{n_m}\mathbf{1}_{n_m}^\top)$, $m>1$, $\HH_2=\mathbf{I}_n$ and $\tht_0=(\sigma^2_v,\sigma^2_e)^\top$. 
Here $m=\min_{k\in\{1,\ldots,r\}}\mbox{rk}(\HH_k)$ translates to the number of clusters in the LMM. 
It is well known that the ratio of the entries in $\tht_0$ is
influential for its estimation, see \citep{Kramlinger2018}. In
 this LMM the intraclass correlation ratio 
\begin{equation*}
\gamma_i = \frac{\sigma_v^2}{\sigma_v^2+\sigma_e^2/n_i}
\end{equation*}
measures the class variability relative to the total variability. If
$\gamma_i \approx 0$, $\sigma_v^2$ cannot be estimated reliably.
Condition (\ref{c_theta}) translates into a positive lower bound for
$\gamma_i$. Further, $\HH_1\geq0$ and $\HH_2>0$, which corresponds to
condition (\ref{c_linear}). Condition (\ref{c_x}) is trivially
satisfied. Next, (\ref{c_h}) and (\ref{c_pd}) can be verified
directly, since all quantities involved are known. For, e.g., $c=1$,
$n_1 = \dots = n_m = 2$, this gives $\underline\omega = 1 / 3$ and
$\overline\omega = 2 / 3$ and
\begin{align*}
\KK = \frac{1}{9}\begin{pmatrix} {4}\frac{m-1}{m}& {\sqrt{2}}\frac{m-1}{m}\\
 {\sqrt{2}}\frac{m-1}{m}& \frac{10m-1}{2m}\end{pmatrix},
\end{align*}
for which $b = 1/5$. Note that condition (\ref{c_h}) includes cases in
which $n_i$ is bounded as well as cases where $n_i$ is unbounded, $i \in\{1,\ldots,m\}$. This has implications at which rate $\tht_0$ and in turn
$\bta_0$ are estimated. For instance, the rate in which $\sigma_v^2$
is estimated via REML is governed by $m$, rather than the total number
of observations $n = \sum_{i=1}^m n_i$. Clearly, if all $n_i$, $i \in\{1,\ldots,m\}$ are bounded, $m = O(n)$. Since (\ref{c_h}) holds for both
asymptotic scenarios and for balanced ($n_1 = \dots = n_m$) as well as
unbalanced ($n_i \neq n_j$) LMMs, the condition is sufficiently
general for most applications.
\end{exa}
\begin{exa} 
Equation (\ref{GenLM}) allows for more general covariance
structures than  (\ref{LMM}) in which the observations are not
necessarily clustered as independent blocks. Assume a model with two
random effects $\vv\in\mathbb{R}^m$ and
$\tilde\vv\in\mathbb{R}^{\tilde{m}}$
\begin{equation*}
\y = \X\bta_0 + \Z \vv + \tilde\Z\tilde\vv +  \veeps, 
\end{equation*}
where $\Z = \text{blockdiag}(\mathbf{1}_{n_i})_{i=1}^m$, $\tilde\Z
= \text{blockdiag}(\mathbf{1}_{\tilde{n}_j})_{j=1}^{\tilde{m}}$. Condition (\ref{c_linear}) is fulfilled if $\HH_1 = \Z\Z^\top$ and $\HH_2 = \tilde\Z\tilde\Z^\top$. However, if $m = \tilde{m}$ and $n_i = \tilde{n}_i$ for all $i \in\{1,\ldots,m\}$, (\ref{c_pd}) is not satisfied. A variation of such a model is discussed in Section \ref{sec:app}.
\end{exa}
 


\section{Uniform Cram\'er consistency of the covariance parameters}
\label{sec:reml} %
The key idea in deriving uniformly valid confidence sets in LMMs is
that the estimation of $\bta_0$ and $\tht_0$ can be separated. The
virtue of restricted maximum likelihood (REML) lies in estimating the
covariance parameters, while accounting for the loss of degrees of
freedom from estimation of $\bta_0$ \citep{Searle1992}. 
The REML estimator for $\tht_0$ is a maximizer of
\begin{equation}
\begin{aligned}\label{REML}
\ell_R(\tht) = -\frac{1}{2}\ln|\V(\tht)|-\frac{1}{2}\ln|\X^\top\V(\tht)^{-1}\X|-\frac{1}{2}\y^\top\PP(\tht)\y,
\end{aligned}
\end{equation}
where $\PP(\tht) = \PP(\tht,1)$, as defined in (\ref{c_pd})
\citep{LaMotte07}. By construction, $\ell_R$ is constant with respect
to $\bta_0$ as $\PP(\tht)\X\bta_0 = \boldsymbol{0}_{n}$. 
%

As interest lies in uniformly valid inference, uniform consistency of
the REML estimator for $\tht_0$ is required. Consistency of general
maximum likelihood estimators (MLE) was first shown by
\citep{Wald1949}, and uniform MLE consistency by \citep{Moran1970}. Both
required the parameter space to be compact, as well as independently
drawn observations. %
A consistency (but not uniform consistency) result that relaxes those
assumptions was first given by \citep{Weiss1971, Weiss1973}.
\citep{Miller1977} extended these results for MLEs
in LMMs, while \citep{Das1979} and \citep{Cressie1993} did so for REML
estimators. \citep{Jiang1996} considered the REML estimator under
non-normal LMMs and for unbounded $p$. However, to the best of our
knowledge, uniform consistency of REML estimators in LMMs has not been
treated.

%
In the following theorem, we prove the {Cram\'er}-type uniform consistency over $\ThetaC$ for REML estimators in LMMs. That is, the uniform consistency is shown for a sequence of local maximizers of 
(\ref{REML}), but this sequence does not have to be unique.
In fact, for unbalanced panels in
LMMs, that is if $n_i \neq n_j$, for $i \neq j$ in (\ref{LMM}), it cannot
be guaranteed that a maximizer of (\ref{REML}) is unique
\citep{Jiang1997}. %

\begin{thm} \label{unifconv}
Let model (\ref{GenLM}) and (\ref{c_theta}) - (\ref{c_pd}) hold and let
$\nu_k(\tht_0) = \rk(\HH_k)^{-1/2}\tr\{\V(\tht_0)^{-1}\HH_k \}$ for
$k\in\{1,\dots, r\}$ and $m = \min_{k\in\{1,\dots, r\}}\{\rk(\HH_k)\}\rightarrow\infty$.
Then, there exists a sequence $\hat\tht$ of local maximizers of (\ref{REML}), such that
\begin{equation*}
\nu_k(\tht_0) \left|\hat{\theta}_k-\theta_{0,k}\right| = O_P(1)
\end{equation*}
uniformly over $\tht_0\in\ThetaC$.
\end{thm}
Throughout the paper, the notation $Z_m = O_P(a)$ for a sequence of
random variables $Z_m$, $m \in \mathbb{N}$, is understood to hold
{uniformly}. More concretely, it is defined that for all
$\varepsilon>0$, there exist constants $M_\varepsilon, N_\varepsilon$
such that for all $m>N_\varepsilon$:
$\inf_{\tht_0\in\Theta}\Prob_{\tht_0}(Z_m\leq M_\varepsilon) \geq 1 -
\varepsilon$. %
Regarding the normalizing sequence $\nu_k(\tht_0)$, note that
$-\mbox{E}_{\tht_0}\{\partial^2\ell_R(\tht)/\partial\theta_k^2|_{\tht_0}\} = O\{\nu_k(\tht_0)^2\}$, see (\ref{trsplitREML}) for details. 
The proof of Theorem~\ref{unifconv} follows the lines of
\citep{Moran1970, Weiss1971}. We extend their reasoning to
account for a uniformity argument over $\Theta$ by the construction of the
normalizing sequence $\nu_k(\tht_0)$. If $\tht_0$ were fixed, the
normalizing sequence would reduce to $\sqrt{\rk(\HH_k)}$, which
corresponds to the parametric rate in classical linear models
\citep{Jiang1996}. However, the specific choice of $\nu_k(\tht_0)$
allows to prove uniform consistency with the help of the auxiliary
result given in Lemma~\ref{Qlemma} in Section~\ref{sec:proofs}.

\section{Uniform  confidence regions for the coefficient parameters}
\label{sec:main_results}
Now we can derive a uniform confidence region for the parameter $\boldsymbol{\beta}_0$, 
based on the Lasso estimator. 
For linear models, it is known that the minimal (infimal) coverage
for Lasso-based confidence sets over the space of coefficient
parameters occurs when the parameter values are large in absolute
value, depending only on the sign of the entries of $\bta_0$. This
finding carries over to infimal coverage over the space of coefficient
and covariance parameters, as the latter, when estimated with REML, do
not depend on the fixed effects.
\begin{prop} \label{approxlimcase}
Let model (\ref{GenLM}) and  (\ref{c_theta}) - (\ref{c_pd}) hold. Then
there exists a sequence $\hat\tht$ of local maximizers of
(\ref{REML}) such that
\begin{equation*}
\inf_{\substack{\bta_0\in\mathbb{R}^p\\\tht_0\in\ThetaC}}\Prob_{\bta_0,\tht_0}\left\{ \sqrt{n}\left( \hat\bta_L - \bta_0 \right) \in E\left(\hat\C, t \right) \right\}
	= \inf_{\substack{\dd\in\{-1,1\}^p\\\tht_0\in\ThetaC}}\Prob_{\tht_0}\left\{ \hat\uu_\dd \in E\left(\hat\C,t\right) \right\}.
\end{equation*}
\end{prop}
The proof follows from \cite[Theorem 1]{Ewald2018} as $\hat\tht$
is independent from $\bta_0$ by construction. Instead of evaluating
the infimum for $\bta_0$ directly, the result states that a much
simpler minimization over a discrete set is sufficient to evaluate the
probability. %
In the context of LMMs, the key consequence of
Proposition~\ref{approxlimcase} is that, due to the orthogonality of
$\hat\tht$ and $\bta_0$, the uniform coverage probability can be
separately evaluated for $\bta_0$ and $\tht_0$. For the latter,
Theorem~\ref{unifconv} allows to replace $\hat\tht$ by $\tht_0$ at
a cost of an error term with a tractable dependency on $\tht_0$.
Eventually, this leads to the following result.
\begin{thm}  \label{mainres}
Let model (\ref{GenLM}) and (\ref{c_theta})-(\ref{c_pd}) hold and $m = \min_{k\in\{1,\dots, r\}}\{\rk(\HH_k)\}\rightarrow\infty$.
Then, there exists a sequence $\hat\tht$ of local maximizers of (\ref{REML}) such that 
\begin{equation*}
\hat\tau =   \max_{\dd\in\{-1,1\}^p} \chi^2_{p, 1-\alpha} \left(  \left\| \hat\C^{-1/2}\LAM\dd  \right\|^2    \right)
\end{equation*}
the quantile of the non-central $\chi^2_p$-distribution, and it holds for all $\tht_0\in\ThetaC$ that
\begin{equation*}
\inf_{\substack{\bta_0\in\mathbb{R}^p\\\tht_0\in\ThetaC}}\Prob_{\bta_0,\tht_0}\left\{ \sqrt{n}\left( \hat\bta_L - \bta_0 \right) \in E\left(\hat\C, \hat\tau  \right) \right\}
	= 1 - \alpha + O\left(\frac{1}{\sqrt{m}}\right).
\end{equation*}
\end{thm}

The result differs from Proposition~\ref{approxlimcase} in two
aspects. First, the minimization over $\dd\in\{-1,1\}^p$ has been
shifted to the choice of the non-centrality parameter of the
$\chi^2$-distribution, based on Propositions 4 and 5 from
\citep{Ewald2018}. 
Second, the estimation of $\tht_0$ only contributes to an error term
of uniformly vanishing,  parametric order $m^{-1/2}$. The rate is
given by the $m$ observations that contribute to the estimation of
$\theta_{0,k'}$, $k' = \argmin_{k\in\{1,\ldots,r\}} \rk(\HH_k)$. The crucial finding is
that the error terms that occur in the estimation of both $\bta_0$ and
$\tht_0$ and which depend on $\tht_0$ eventually translate to an error
term uniformly independent of $\tht_0$. As a consequence of this
result, it follows that the Lasso-based confidence ellipsoid
\begin{equation} \label{ConfSet}
M = \left\{ \bta \in \mathbb{R}^p\text{:}\quad  n \left\| \hat\C^{1/2} \left( \hat\bta_L - \bta\right) \right\|^2 \leq \hat\tau \right\}
\end{equation}
attains nominal coverage uniformly up to an error of a parametric
rate. %
By separating the estimation of $\bta_0$ with penalization and
$\tht_0$ without it, the result allows to infer about the fixed
effects for data with a complex dependency structure. From the
orthogonality argument, it is clear that the penalization on $\bta_0$
does not interfere with inference based on $\hat\tht$, for which
existing theory can be applied. Since the REML estimates exhibit no
analytically tractable finite sample distribution, the approximation
with a $\chi^2$-quantile is only valid asymptotically. Although this
implies that the resulting testing procedure is not of nominal level
$1-\alpha$ as discussed in \citep{Leeb2017}, the simulation in
Section~\ref{sec:sims} suggests that this error term has little
influence in finite samples. Most importantly, however, is that the
error term does not depend on $\bta_0$ and $\tht_0$ and the stated
coverage probability is attained uniformly over both the space of
regression coefficients and covariance parameters. Nevertheless, the
precise coverage probability depends on the sign of the entries of
$\bta_0$. Specifically, all but two orthants exhibit an overcoverage.
This effect is visualized and discussed in Fig. \ref{quad_fig} in
Section~\ref{sec:sims}.

When compared to standard confidence sets based on estimation
procedures without variable selection, it is well-known that the
confidence set $M$ is larger \citep{Poetscher2009}. In many
applications, however, smaller models and thus selection of covariates
is preferred. In such cases, $\hat\bta_L$ and inference based upon
it offer an alternative, despite the Lasso's difficult distributional
properties.

\section{Simulations} \label{sec:sims}

In this section, the finite sample coverage probabilities of the
confidence set  (\ref{ConfSet}) are compared with several
alternatives. First, with a least-squares-based approach, which offers
uniform coverage but no selection. Second, with a least-squares-based
approach after selecting a model with AIC, resulting in type-I-error
inflation. And finally, we demonstrate that a least-squares-based
confidence set centered at $\hat\bta_L$ does not meet the nominal
level. For visualization, we use the following simulation design for
two coefficient parameters. Consider the `random intercept model', a
special case of the model (\ref{LMM}),
\begin{equation}
\begin{gathered}\label{RIM}
y_{ij} = \x_{ij}^\top\bta_0 + v_i + u_{ij},\quad i\in\{1,\ldots,m\}, \quad j\in\{1,\dots, n_i\}, \\ 
u_{ij} \overset{i.i.d.} \sim \mathcal{N}(0,\sigma_u^2), \quad v_i \overset{i.i.d.} \sim \mathcal{N}(0,\sigma_v^2).
\end{gathered}
\end{equation}
We fix the parameters to $m=20$, $n=400$, $n_i=20$, $\sigma_u=\sigma_v=4$. Note that the parameter $m$ that determines the convergence rate is quite small. 
For visualization purposes, we restrict the simulation to $p=2$ parameters.
The tuning parameters are chosen to be $\lambda_i = n^{1/2}/2$ for $i\in\{1,2\}$, which corresponds to a conservative tuning regime.
The entries $\x_{ij} = (x_{1ij}, x_{2ij})^\top$ of the matrix of covariables are independently drawn once from $\mathcal{N}(0,4)$ so that the fixed and random effects are of a comparable magnitude.
For each $\bta_0\in[-4,4]\times[-4,4]$ on an equidistant grid of $81\times81$, $2$,$000$ Monte Carlo simulations were carried out.
For each simulation, the empirical coverage probability is computed by checking if $\bta_0\in M$, for $M$ with $\alpha = 0.05$ from (\ref{ConfSet}).
Fig.~\ref{quad_fig} shows the results.

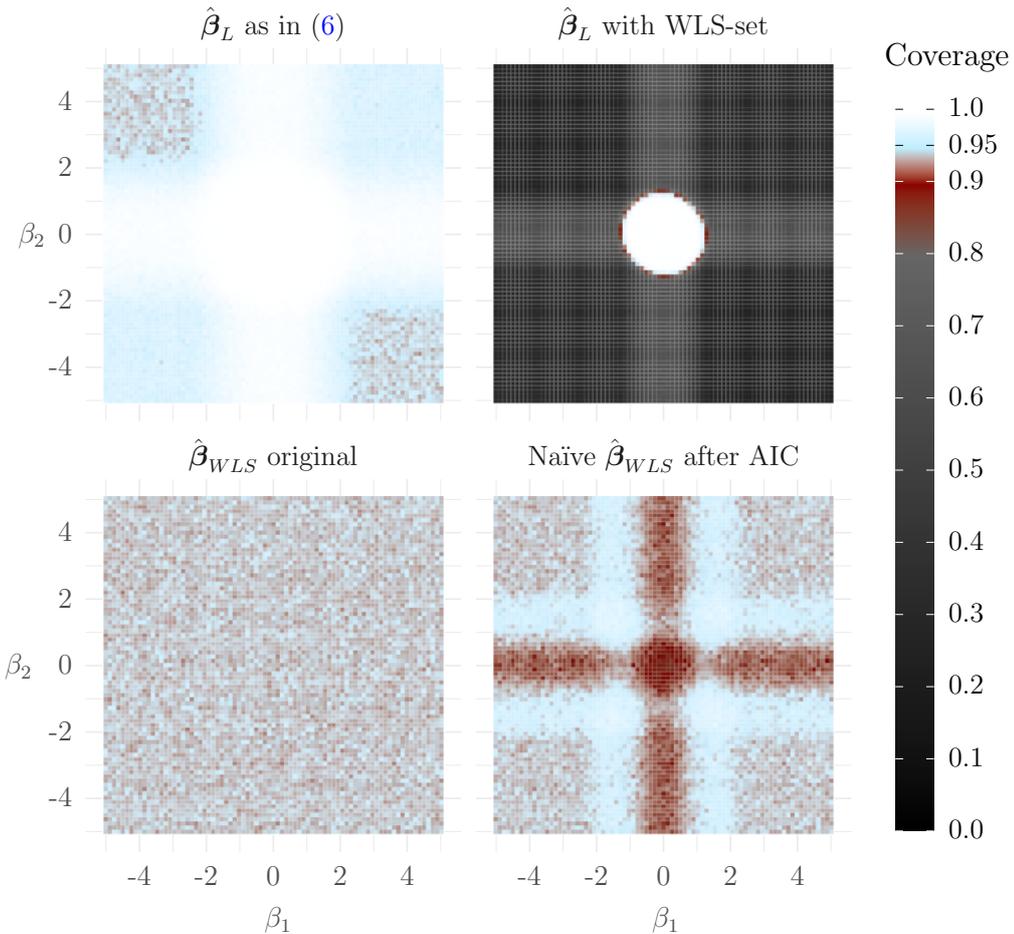
\begin{figure}[t]
	\hspace*{1.8cm}
\begin{tikzpicture}[x=1pt,y=1pt]
\definecolor{fillColor}{RGB}{255,255,255}
\path[use as bounding box,fill=fillColor,fill opacity=0.00] (0,0) rectangle (375.80,375.80);
\begin{scope}
\path[clip] ( 30.41,203.33) rectangle (171.01,343.93);
\definecolor{drawColor}{gray}{0.92}

\path[draw=drawColor,line width= 0.3pt,line join=round] ( 30.41,210.98) --
	(171.01,210.98);

\path[draw=drawColor,line width= 0.3pt,line join=round] ( 30.41,236.04) --
	(171.01,236.04);

\path[draw=drawColor,line width= 0.3pt,line join=round] ( 30.41,261.10) --
	(171.01,261.10);

\path[draw=drawColor,line width= 0.3pt,line join=round] ( 30.41,286.16) --
	(171.01,286.16);

\path[draw=drawColor,line width= 0.3pt,line join=round] ( 30.41,311.23) --
	(171.01,311.23);

\path[draw=drawColor,line width= 0.3pt,line join=round] ( 30.41,336.29) --
	(171.01,336.29);

\path[draw=drawColor,line width= 0.3pt,line join=round] ( 38.06,203.33) --
	( 38.06,343.93);

\path[draw=drawColor,line width= 0.3pt,line join=round] ( 63.12,203.33) --
	( 63.12,343.93);

\path[draw=drawColor,line width= 0.3pt,line join=round] ( 88.18,203.33) --
	( 88.18,343.93);

\path[draw=drawColor,line width= 0.3pt,line join=round] (113.24,203.33) --
	(113.24,343.93);

\path[draw=drawColor,line width= 0.3pt,line join=round] (138.30,203.33) --
	(138.30,343.93);

\path[draw=drawColor,line width= 0.3pt,line join=round] (163.37,203.33) --
	(163.37,343.93);

\path[draw=drawColor,line width= 0.6pt,line join=round] ( 30.41,223.51) --
	(171.01,223.51);

\path[draw=drawColor,line width= 0.6pt,line join=round] ( 30.41,248.57) --
	(171.01,248.57);

\path[draw=drawColor,line width= 0.6pt,line join=round] ( 30.41,273.63) --
	(171.01,273.63);

\path[draw=drawColor,line width= 0.6pt,line join=round] ( 30.41,298.69) --
	(171.01,298.69);

\path[draw=drawColor,line width= 0.6pt,line join=round] ( 30.41,323.76) --
	(171.01,323.76);

\path[draw=drawColor,line width= 0.6pt,line join=round] ( 50.59,203.33) --
	( 50.59,343.93);

\path[draw=drawColor,line width= 0.6pt,line join=round] ( 75.65,203.33) --
	( 75.65,343.93);

\path[draw=drawColor,line width= 0.6pt,line join=round] (100.71,203.33) --
	(100.71,343.93);

\path[draw=drawColor,line width= 0.6pt,line join=round] (125.77,203.33) --
	(125.77,343.93);

\path[draw=drawColor,line width= 0.6pt,line join=round] (150.84,203.33) --
	(150.84,343.93);
\node[inner sep=0pt,outer sep=0pt,anchor=south west,rotate=  0.00] at ( 36.80, 209.72) {
	\pgfimage[width=127.82pt,height=127.82pt,interpolate=false]{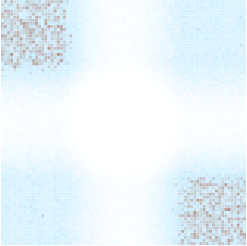}};
\end{scope}
\begin{scope}
\path[clip] ( 30.41, 40.66) rectangle (171.01,181.26);
\definecolor{drawColor}{gray}{0.92}

\path[draw=drawColor,line width= 0.3pt,line join=round] ( 30.41, 48.31) --
	(171.01, 48.31);

\path[draw=drawColor,line width= 0.3pt,line join=round] ( 30.41, 73.37) --
	(171.01, 73.37);

\path[draw=drawColor,line width= 0.3pt,line join=round] ( 30.41, 98.43) --
	(171.01, 98.43);

\path[draw=drawColor,line width= 0.3pt,line join=round] ( 30.41,123.49) --
	(171.01,123.49);

\path[draw=drawColor,line width= 0.3pt,line join=round] ( 30.41,148.56) --
	(171.01,148.56);

\path[draw=drawColor,line width= 0.3pt,line join=round] ( 30.41,173.62) --
	(171.01,173.62);

\path[draw=drawColor,line width= 0.3pt,line join=round] ( 38.06, 40.66) --
	( 38.06,181.26);

\path[draw=drawColor,line width= 0.3pt,line join=round] ( 63.12, 40.66) --
	( 63.12,181.26);

\path[draw=drawColor,line width= 0.3pt,line join=round] ( 88.18, 40.66) --
	( 88.18,181.26);

\path[draw=drawColor,line width= 0.3pt,line join=round] (113.24, 40.66) --
	(113.24,181.26);

\path[draw=drawColor,line width= 0.3pt,line join=round] (138.30, 40.66) --
	(138.30,181.26);

\path[draw=drawColor,line width= 0.3pt,line join=round] (163.37, 40.66) --
	(163.37,181.26);

\path[draw=drawColor,line width= 0.6pt,line join=round] ( 30.41, 60.84) --
	(171.01, 60.84);

\path[draw=drawColor,line width= 0.6pt,line join=round] ( 30.41, 85.90) --
	(171.01, 85.90);

\path[draw=drawColor,line width= 0.6pt,line join=round] ( 30.41,110.96) --
	(171.01,110.96);

\path[draw=drawColor,line width= 0.6pt,line join=round] ( 30.41,136.02) --
	(171.01,136.02);

\path[draw=drawColor,line width= 0.6pt,line join=round] ( 30.41,161.09) --
	(171.01,161.09);

\path[draw=drawColor,line width= 0.6pt,line join=round] ( 50.59, 40.66) --
	( 50.59,181.26);

\path[draw=drawColor,line width= 0.6pt,line join=round] ( 75.65, 40.66) --
	( 75.65,181.26);

\path[draw=drawColor,line width= 0.6pt,line join=round] (100.71, 40.66) --
	(100.71,181.26);

\path[draw=drawColor,line width= 0.6pt,line join=round] (125.77, 40.66) --
	(125.77,181.26);

\path[draw=drawColor,line width= 0.6pt,line join=round] (150.84, 40.66) --
	(150.84,181.26);
\node[inner sep=0pt,outer sep=0pt,anchor=south west,rotate=  0.00] at ( 36.80,  47.05) {
	\pgfimage[width=127.82pt,height=127.82pt,interpolate=false]{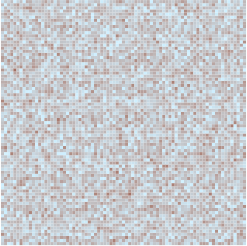}};
\end{scope}
\begin{scope}
\path[clip] (176.51,203.33) rectangle (317.11,343.93);
\definecolor{drawColor}{gray}{0.92}

\path[draw=drawColor,line width= 0.3pt,line join=round] (176.51,210.98) --
	(317.11,210.98);

\path[draw=drawColor,line width= 0.3pt,line join=round] (176.51,236.04) --
	(317.11,236.04);

\path[draw=drawColor,line width= 0.3pt,line join=round] (176.51,261.10) --
	(317.11,261.10);

\path[draw=drawColor,line width= 0.3pt,line join=round] (176.51,286.16) --
	(317.11,286.16);

\path[draw=drawColor,line width= 0.3pt,line join=round] (176.51,311.23) --
	(317.11,311.23);

\path[draw=drawColor,line width= 0.3pt,line join=round] (176.51,336.29) --
	(317.11,336.29);

\path[draw=drawColor,line width= 0.3pt,line join=round] (184.15,203.33) --
	(184.15,343.93);

\path[draw=drawColor,line width= 0.3pt,line join=round] (209.22,203.33) --
	(209.22,343.93);

\path[draw=drawColor,line width= 0.3pt,line join=round] (234.28,203.33) --
	(234.28,343.93);

\path[draw=drawColor,line width= 0.3pt,line join=round] (259.34,203.33) --
	(259.34,343.93);

\path[draw=drawColor,line width= 0.3pt,line join=round] (284.40,203.33) --
	(284.40,343.93);

\path[draw=drawColor,line width= 0.3pt,line join=round] (309.46,203.33) --
	(309.46,343.93);

\path[draw=drawColor,line width= 0.6pt,line join=round] (176.51,223.51) --
	(317.11,223.51);

\path[draw=drawColor,line width= 0.6pt,line join=round] (176.51,248.57) --
	(317.11,248.57);

\path[draw=drawColor,line width= 0.6pt,line join=round] (176.51,273.63) --
	(317.11,273.63);

\path[draw=drawColor,line width= 0.6pt,line join=round] (176.51,298.69) --
	(317.11,298.69);

\path[draw=drawColor,line width= 0.6pt,line join=round] (176.51,323.76) --
	(317.11,323.76);

\path[draw=drawColor,line width= 0.6pt,line join=round] (196.69,203.33) --
	(196.69,343.93);

\path[draw=drawColor,line width= 0.6pt,line join=round] (221.75,203.33) --
	(221.75,343.93);

\path[draw=drawColor,line width= 0.6pt,line join=round] (246.81,203.33) --
	(246.81,343.93);

\path[draw=drawColor,line width= 0.6pt,line join=round] (271.87,203.33) --
	(271.87,343.93);

\path[draw=drawColor,line width= 0.6pt,line join=round] (296.93,203.33) --
	(296.93,343.93);
\node[inner sep=0pt,outer sep=0pt,anchor=south west,rotate=  0.00] at (182.90, 209.72) {
	\pgfimage[width=127.82pt,height=127.82pt,interpolate=false]{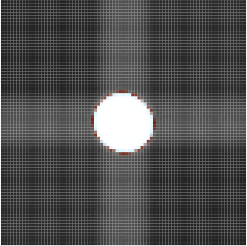}};
\end{scope}
\begin{scope}
\path[clip] (176.51, 40.66) rectangle (317.11,181.26);
\definecolor{drawColor}{gray}{0.92}

\path[draw=drawColor,line width= 0.3pt,line join=round] (176.51, 48.31) --
	(317.11, 48.31);

\path[draw=drawColor,line width= 0.3pt,line join=round] (176.51, 73.37) --
	(317.11, 73.37);

\path[draw=drawColor,line width= 0.3pt,line join=round] (176.51, 98.43) --
	(317.11, 98.43);

\path[draw=drawColor,line width= 0.3pt,line join=round] (176.51,123.49) --
	(317.11,123.49);

\path[draw=drawColor,line width= 0.3pt,line join=round] (176.51,148.56) --
	(317.11,148.56);

\path[draw=drawColor,line width= 0.3pt,line join=round] (176.51,173.62) --
	(317.11,173.62);

\path[draw=drawColor,line width= 0.3pt,line join=round] (184.15, 40.66) --
	(184.15,181.26);

\path[draw=drawColor,line width= 0.3pt,line join=round] (209.22, 40.66) --
	(209.22,181.26);

\path[draw=drawColor,line width= 0.3pt,line join=round] (234.28, 40.66) --
	(234.28,181.26);

\path[draw=drawColor,line width= 0.3pt,line join=round] (259.34, 40.66) --
	(259.34,181.26);

\path[draw=drawColor,line width= 0.3pt,line join=round] (284.40, 40.66) --
	(284.40,181.26);

\path[draw=drawColor,line width= 0.3pt,line join=round] (309.46, 40.66) --
	(309.46,181.26);

\path[draw=drawColor,line width= 0.6pt,line join=round] (176.51, 60.84) --
	(317.11, 60.84);

\path[draw=drawColor,line width= 0.6pt,line join=round] (176.51, 85.90) --
	(317.11, 85.90);

\path[draw=drawColor,line width= 0.6pt,line join=round] (176.51,110.96) --
	(317.11,110.96);

\path[draw=drawColor,line width= 0.6pt,line join=round] (176.51,136.02) --
	(317.11,136.02);

\path[draw=drawColor,line width= 0.6pt,line join=round] (176.51,161.09) --
	(317.11,161.09);

\path[draw=drawColor,line width= 0.6pt,line join=round] (196.69, 40.66) --
	(196.69,181.26);

\path[draw=drawColor,line width= 0.6pt,line join=round] (221.75, 40.66) --
	(221.75,181.26);

\path[draw=drawColor,line width= 0.6pt,line join=round] (246.81, 40.66) --
	(246.81,181.26);

\path[draw=drawColor,line width= 0.6pt,line join=round] (271.87, 40.66) --
	(271.87,181.26);

\path[draw=drawColor,line width= 0.6pt,line join=round] (296.93, 40.66) --
	(296.93,181.26);
\node[inner sep=0pt,outer sep=0pt,anchor=south west,rotate=  0.00] at (182.90,  47.05) {
	\pgfimage[width=127.82pt,height=127.82pt,interpolate=false]{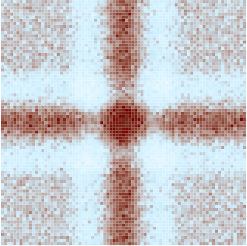}};
\end{scope}
\begin{scope}
\path[clip] ( 30.41,181.26) rectangle (171.01,197.83);
\definecolor{drawColor}{gray}{0.10}

\node[text=drawColor,anchor=base,inner sep=0pt, outer sep=0pt, scale=  0.88] at (100.71,186.52) {$\hat\bta_{WLS}$ original};
\end{scope}
\begin{scope}
\path[clip] (176.51,181.26) rectangle (317.11,197.83);
\definecolor{drawColor}{gray}{0.10}

\node[text=drawColor,anchor=base,inner sep=0pt, outer sep=0pt, scale=  0.88] at (246.81,186.52) {Na\"ive $\hat\bta_{WLS}$ after AIC};
\end{scope}
\begin{scope}
\path[clip] ( 30.41,343.93) rectangle (171.01,360.50);
\definecolor{drawColor}{gray}{0.10}

\node[text=drawColor,anchor=base,inner sep=0pt, outer sep=0pt, scale=  0.88] at (100.71,349.19) {$\hat\bta_{L}$ as in (\ref{ConfSet})};
\end{scope}
\begin{scope}
\path[clip] (176.51,343.93) rectangle (317.11,360.50);
\definecolor{drawColor}{gray}{0.10}

\node[text=drawColor,anchor=base,inner sep=0pt, outer sep=0pt, scale=  0.88] at (246.81,349.19) {$\hat\bta_L$ with WLS-set};
\end{scope}
\begin{scope}
\path[clip] (  0.00,  0.00) rectangle (375.80,375.80);
\definecolor{drawColor}{gray}{0.30}

\node[text=drawColor,anchor=base,inner sep=0pt, outer sep=0pt, scale=  0.90] at ( 50.59, 28.29) {-4};

\node[text=drawColor,anchor=base,inner sep=0pt, outer sep=0pt, scale=  0.90] at ( 75.65, 28.29) {-2};

\node[text=drawColor,anchor=base,inner sep=0pt, outer sep=0pt, scale=  0.90] at (100.71, 28.29) {0};
\node[text=drawColor,anchor=base east,inner sep=0pt, outer sep=0pt, scale=  0.90] at (107.71, 12.29) {$\beta_1$};

\node[text=drawColor,anchor=base,inner sep=0pt, outer sep=0pt, scale=  0.90] at (125.77, 28.29) {2};

\node[text=drawColor,anchor=base,inner sep=0pt, outer sep=0pt, scale=  0.90] at (150.84, 28.29) {4};
\end{scope}
\begin{scope}
\path[clip] (  0.00,  0.00) rectangle (375.80,375.80);
\definecolor{drawColor}{gray}{0.30}

\node[text=drawColor,anchor=base,inner sep=0pt, outer sep=0pt, scale=  0.90] at (196.69, 28.29) {-4};

\node[text=drawColor,anchor=base,inner sep=0pt, outer sep=0pt, scale=  0.90] at (221.75, 28.29) {-2};

\node[text=drawColor,anchor=base,inner sep=0pt, outer sep=0pt, scale=  0.90] at (246.81, 28.29) {0};
\node[text=drawColor,anchor=base east,inner sep=0pt, outer sep=0pt, scale=  0.90] at (252.81, 12.29) {$\beta_1$};

\node[text=drawColor,anchor=base,inner sep=0pt, outer sep=0pt, scale=  0.90] at (271.87, 28.29) {2};

\node[text=drawColor,anchor=base,inner sep=0pt, outer sep=0pt, scale=  0.90] at (296.93, 28.29) {4};
\end{scope}
\begin{scope}
\path[clip] (  0.00,  0.00) rectangle (375.80,375.80);
\definecolor{drawColor}{gray}{0.30}

\node[text=drawColor,anchor=base east,inner sep=0pt, outer sep=0pt, scale=  0.90] at ( 25.46,220.41) {-4};

\node[text=drawColor,anchor=base east,inner sep=0pt, outer sep=0pt, scale=  0.90] at ( 25.46,245.47) {-2};

\node[text=drawColor,anchor=base east,inner sep=0pt, outer sep=0pt, scale=  0.90] at ( 25.46,270.53) {0};
\node[text=drawColor,anchor=base east,inner sep=0pt, outer sep=0pt, scale=  0.90] at ( 15.46,270.53) {$\beta_2$};

\node[text=drawColor,anchor=base east,inner sep=0pt, outer sep=0pt, scale=  0.90] at ( 25.46,295.59) {2};

\node[text=drawColor,anchor=base east,inner sep=0pt, outer sep=0pt, scale=  0.90] at ( 25.46,320.66) {4};
\end{scope}
\begin{scope}
\path[clip] (  0.00,  0.00) rectangle (375.80,375.80);
\definecolor{drawColor}{gray}{0.30}

\node[text=drawColor,anchor=base east,inner sep=0pt, outer sep=0pt, scale=  0.90] at ( 25.46, 57.74) {-4};

\node[text=drawColor,anchor=base east,inner sep=0pt, outer sep=0pt, scale=  0.90] at ( 25.46, 82.80) {-2};

\node[text=drawColor,anchor=base east,inner sep=0pt, outer sep=0pt, scale=  0.90] at ( 25.46,107.86) {0};
\node[text=drawColor,anchor=base east,inner sep=0pt, outer sep=0pt, scale=  0.90] at ( 10.46,107.86) {$\beta_2$};

\node[text=drawColor,anchor=base east,inner sep=0pt, outer sep=0pt, scale=  0.90] at ( 25.46,132.93) {2};

\node[text=drawColor,anchor=base east,inner sep=0pt, outer sep=0pt, scale=  0.90] at ( 25.46,157.99) {4};
\end{scope}
\begin{scope}
\path[clip] (  0.00,  0.00) rectangle (375.80,375.80);
\node[inner sep=0pt,outer sep=0pt,anchor=south west,rotate=  0.00] at (333.61, 48.56) {
	\pgfimage[width= 14.45pt,height= 272.27pt,interpolate=true]{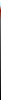}};
\end{scope}
\begin{scope}
\path[clip] (  0.00,  0.00) rectangle (375.80,375.80);
\definecolor{drawColor}{RGB}{0,0,0}

\node[text=drawColor,anchor=base west,inner sep=0pt, outer sep=0pt, scale=  0.88] at (353.56, 48.56 + 9.9 * 27.227) {1.0};
\node[text=drawColor,anchor=base west,inner sep=0pt, outer sep=0pt, scale=  0.88] at (353.56, 48.56 + 9.39 * 27.227) {0.95};
\node[text=drawColor,anchor=base west,inner sep=0pt, outer sep=0pt, scale=  0.88] at (353.56, 48.56 + 8.9 * 27.227) {0.9};
\node[text=drawColor,anchor=base west,inner sep=0pt, outer sep=0pt, scale=  0.88] at (353.56, 48.56 + 7.9 * 27.227) {0.8};
\node[text=drawColor,anchor=base west,inner sep=0pt, outer sep=0pt, scale=  0.88] at (353.56, 48.56 + 6.9 * 27.227) {0.7};
\node[text=drawColor,anchor=base west,inner sep=0pt, outer sep=0pt, scale=  0.88] at (353.56, 48.56 + 5.9 * 27.227) {0.6};
\node[text=drawColor,anchor=base west,inner sep=0pt, outer sep=0pt, scale=  0.88] at (353.56, 48.56 + 4.9 * 27.227) {0.5};
\node[text=drawColor,anchor=base west,inner sep=0pt, outer sep=0pt, scale=  0.88] at (353.56, 48.56 + 3.9 * 27.227) {0.4};
\node[text=drawColor,anchor=base west,inner sep=0pt, outer sep=0pt, scale=  0.88] at (353.56, 48.56 + 2.9 * 27.227) {0.3};
\node[text=drawColor,anchor=base west,inner sep=0pt, outer sep=0pt, scale=  0.88] at (353.56, 48.56 + 1.9 * 27.227) {0.2};
\node[text=drawColor,anchor=base west,inner sep=0pt, outer sep=0pt, scale=  0.88] at (353.56, 48.56 + 0.9 * 27.227) {0.1};
\node[text=drawColor,anchor=base west,inner sep=0pt, outer sep=0pt, scale=  0.88] at (353.56, 48.56 - 0.1 * 27.227) {0.0};

\end{scope}
\begin{scope}
\path[clip] (  0.00,  0.00) rectangle (375.80,375.80);
\definecolor{drawColor}{RGB}{0,0,0}

\node[text=drawColor,anchor=base west,inner sep=0pt, outer sep=0pt, scale=  1] at (329.61, 337.39) {Coverage};
\end{scope}
\begin{scope}
\path[clip] (  0.00,  0.00) rectangle (375.80,375.80);
\definecolor{drawColor}{RGB}{255,255,255}

\path[draw=drawColor,line width= 0.2pt,line join=round] (333.61, 48.56) -- (336.50, 48.56); 
\path[draw=drawColor,line width= 0.2pt,line join=round] (333.61, 48.56 + 27.227) -- (336.50, 48.56 + 27.227);
\path[draw=drawColor,line width= 0.2pt,line join=round] (333.61, 48.56 + 2 * 27.227) -- (336.50, 48.56 + 2 * 27.227);
\path[draw=drawColor,line width= 0.2pt,line join=round] (333.61, 48.56 + 3 * 27.227) -- (336.50, 48.56 + 3 * 27.227);
\path[draw=drawColor,line width= 0.2pt,line join=round] (333.61, 48.56 + 4 * 27.227) -- (336.50, 48.56 + 4 * 27.227);
\path[draw=drawColor,line width= 0.2pt,line join=round] (333.61, 48.56 + 5 * 27.227) -- (336.50, 48.56 + 5 * 27.227);
\path[draw=drawColor,line width= 0.2pt,line join=round] (333.61, 48.56 + 6 * 27.227) -- (336.50, 48.56 + 6 * 27.227);
\path[draw=drawColor,line width= 0.2pt,line join=round] (333.61, 48.56 + 7 * 27.227) -- (336.50, 48.56 + 7 * 27.227);
\path[draw=drawColor,line width= 0.2pt,line join=round] (333.61, 48.56 + 8 * 27.227) -- (336.50, 48.56 + 8 * 27.227);
\path[draw=drawColor,line width= 0.2pt,line join=round] (333.61, 48.56 + 9 * 27.227) -- (336.50, 48.56 + 9 * 27.227);
\path[draw=black,line width= 0.2pt,line join=round] (333.61, 48.56 + 9.5 * 27.227) -- (336.50, 48.56 + 9.5 * 27.227);
\path[draw=black,line width= 0.2pt,line join=round] (333.61, 48.56 + 10 * 27.227) -- (336.50, 48.56 + 10 * 27.227);

\path[draw=drawColor,line width= 0.2pt,line join=round] (345.17, 48.56) -- (348.06, 48.56); 
\path[draw=drawColor,line width= 0.2pt,line join=round] (345.17, 48.56 + 27.227) -- (348.06, 48.56 + 27.227);
\path[draw=drawColor,line width= 0.2pt,line join=round] (345.17, 48.56 + 2 * 27.227) -- (348.06, 48.56 + 2 * 27.227);
\path[draw=drawColor,line width= 0.2pt,line join=round] (345.17, 48.56 + 3 * 27.227) -- (348.06, 48.56 + 3 * 27.227);
\path[draw=drawColor,line width= 0.2pt,line join=round] (345.17, 48.56 + 4 * 27.227) -- (348.06, 48.56 + 4 * 27.227);
\path[draw=drawColor,line width= 0.2pt,line join=round] (345.17, 48.56 + 5 * 27.227) -- (348.06, 48.56 + 5 * 27.227);
\path[draw=drawColor,line width= 0.2pt,line join=round] (345.17, 48.56 + 6 * 27.227) -- (348.06, 48.56 + 6 * 27.227);
\path[draw=drawColor,line width= 0.2pt,line join=round] (345.17, 48.56 + 7 * 27.227) -- (348.06, 48.56 + 7 * 27.227);
\path[draw=drawColor,line width= 0.2pt,line join=round] (345.17, 48.56 + 8 * 27.227) -- (348.06, 48.56 + 8 * 27.227);
\path[draw=drawColor,line width= 0.2pt,line join=round] (345.17, 48.56 + 9 * 27.227) -- (348.06, 48.56 + 9 * 27.227);
\path[draw=black,line width= 0.2pt,line join=round] (345.17, 48.56 + 9.5 * 27.227) -- (348.06, 48.56 + 9.5 * 27.227);
\path[draw=black,line width= 0.2pt,line join=round] (345.17, 48.56 + 10 * 27.227) -- (348.06, 48.56 + 10 * 27.227);

\end{scope}
\end{tikzpicture}
		\caption{Only confidence sets based on the Lasso as in (\ref{ConfSet}) or WLS estimator (left column) achieve a nominal level over the whole parameter space. The former is conservative around the origin and axes.
Na\"ively applying WLS sets to the Lasso or WLS estimator after AIC variable selection (right column) yields undercoverage. }
		\label{quad_fig}
\end{figure}
The probabilities shown are average empirical coverages for all configurations of $\bta_0$.
First consider the classical confidence sets based on the weighted least-squares (WLS) estimator $\hat\bta_{WLS} = \{\X^\top\V(\hat\tht)^{-1}\X\}^{-1}\X^\top\V(\hat\tht)^{-1}\y$ (bottom left), with no variable selection involved.
As the distribution of $\hat\bta_{WLS} - \bta_0$ is independent of $\bta_0$, the coverage based on its confidence set is attained uniformly over the space of coefficient parameters.
One finds that no deviations from the nominal level of $95\%$ can be observed in the simulation.
Next, consider the confidence set based on the Lasso as given in (\ref{ConfSet}) (upper left).
No undercoverage is observed over the entire parameter space. 
Nominal coverage, up to an error of order $O(m^{-1/2})$ due to the estimation of random effects, is attained for $\text{sign}(\beta_1) = -\text{sign}(\beta_2)$. 
The other orthants exhibit a slight overcoverage (light blue), whereas a significant overcoverage (white) occurs at the axes and around the origin.
The latter effects are due to the event of variable selection when a component in $\hat\bta_L$ is set to zero.
Hence, at the axes, a coverage close to $100\%$ is achieved, and the $95\%$-confidence sets prove to be too wide.
These findings are in line with the example of the linear regression model from \cite[Fig. 4]{Ewald2018}.

The overcoverage in all but two orthants is a consequence of Lemma~\ref{minmization_wrt_d}. 
The confidence sets are based on the minimizers in (\ref{objective_function_limiting}), which depend on $\dd$, the signs of the coefficient parameters. 
The nominal coverage probability is attained at the orthants specified by $\dd^\ast = \argmax_{\dd\in\{-1,1\}^p} \big\| \C^{-1/2}\LAM\dd \big\|^2$, corresponding to the upper left and lower right orthants in Fig. \ref{quad_fig} (upper left). 
In these, the coverage probabilities vary between $93.4\%$ and $95.8\%$, which is close to the nominal level of $95\%$.  
This indicates that the additional uncertainty induced by the estimation of $\tht_0$ is not substantial even if $m$ is small. 
If the coefficient signs of $\beta_1$ and $\beta_2$ differ from $\dd^\ast$, the corresponding confidence sets are conservative. 
In Fig. \ref{quad_fig} (upper left), the coverage probabilities vary between $94.6\%$ and $96.8\%$ in such orthants, which is close to the nominal level, too.  

In contrast to the methods on the left column in Fig.~\ref{quad_fig}, which meet the nominal level thoroughly, but without variable selection (bottom left) or conservatively, but with selection (upper left), two additional na\"ive approaches are displayed on the right column.
First, one observes that na\"ively applying classical WLS confidence sets around a WLS estimator $\hat\bta_{WLS}$ after performing variable selection with AIC (bottom right) yields inconsistent coverages over the parameter space.
Type-I-error inflation \citep{Berk2013} occurs in regions with undercoverage (red) in the event of variable selection at the axes.
Another approach is applying a WLS confidence set around the estimator $\hat\bta_L$ (upper right).
Again, these sets are not theoretically justified and indeed, an overcoverage occurs at the origin, whereas a severe undercoverage over the rest of the coefficient parameter space.
Both na\"ive approaches thus yield misleading confidence sets, and their use is inadvisable.

We have also experimented with more complex covariance matrices, e.g., taking serially dependent observations within a cluster. While this increases the computational burden, the overall conclusion remains the same. 
Similarly, if more than two covariates are considered ($p>2$), 
the confidence sets will
meet nominal level in two out of $2^p$ orthants, and show a slight
overcoverage in the remaining $2^{p-1}$. If $\bta_0$ has close to
zero entries, the confidence sets are very conservative, and the coverage
probabilities close to one.

\section{Application} \label{sec:app}

We demonstrate the performance of the confidence sets based on the Lasso and WLS with no variable selection on a data set on the ecological status of lakes in the northeastern United States collected by the \citep{LakeData}.
The set has been thoroughly investigated in previous studies \citep{Opsomer2007, Breidt2007}.
The model of interest is given by
\begin{equation*}
\begin{gathered}
\y = \mathbf{1}_n\beta_1 + \x\beta_2 + \Z\vv + \mathbf{D}\mathbf{u} + \eeps, \\
\vv\sim\mathcal{N}(\boldsymbol{0}_m, \sigma_v^2\mathbf{I}_m), \quad
\mathbf{u}\sim\mathcal{N}(\boldsymbol{0}_K, \sigma_u^2\mathbf{I}_K), \quad
\eeps\sim\mathcal{N}(\boldsymbol{0}_n, \sigma_e^2\mathbf{I}_n).
\end{gathered}
\end{equation*}
The response variable of interest is the acid-neutralizing capacity in equivalents per liter, which indicates the lake's environmental state.
It is related to a design matrix that includes an intercept and the respective elevation in meters.
In total, the $n=551$ lakes are partitioned via $m=19$ hydrological unit codes with $1 \leq n_i \leq 13$, which enter the model as random effect intercept, i.e.,  $\Z = \text{blockdiag}(\mathbf{1}_{n_i})_{i = 1}^m$.
Finally, a penalized spline with $K=80$ knots accounts for spatial proximity effects.
It is included as a second random effect, for which $\mathbf{D}$ is a transformed radial basis as described by \citep{Opsomer2007}. 

The penalization parameters for the Lasso are set to $\lambda_1 = \lambda_2 = 16.9$ by leave-one-out- cross-validation obtained with the \textsf{R}-function \texttt{glmnet}. This value is close to $\sqrt{n}/2 \approx 11.74$, which corresponds to a conservative tuning regime. 
The fitted model does not indicate any deviation from regularity conditions (\ref{c_theta}) - (\ref{c_pd}).
In particular, $\max(\hat{\tht})/\min(\hat{\tht}) \approx 9.99$.
The estimated fixed effects are given as
\begin{equation*}
  \hat\bta_L = (0, -0.0010)^\top, \quad  \hat\bta_{WLS} = (-0.5589, -0.0011)^\top.
\end{equation*}
In previous studies, the intercept has been deemed insignificant \citep{Opsomer2007}.
Indeed, the estimated standard deviations are given by ${s}(\hat\beta_{WLS,1}) = 2.2130$ and ${s}(\hat\beta_{WLS,2}) = 0.0001$.
Consequently, the Lasso estimator sets the intercept equal to zero.
This corresponds to a selection event on the axes in Fig.~\ref{quad_fig} (upper left).
Such suffers from severe overcoverage.
A $99\%$-confidence set $M_L$ based on the Lasso is given by $M$ from (\ref{ConfSet}), with $\hat\tau = \chi^2_{2, 0.99}(1398) \approx  1578$.
The WLS based set $M_{WLS}$ is given as $M$, but $\hat\bta_L$ replaced with $\hat\bta_{WLS}$ and $\hat\tau$ with $\chi^2_{2,0.99} \approx 9$.
The Lasso-based confidence ellipse is much larger than the WLS-based confidence ellipse.
This does not indicate that the former are useless, however.
Much of their additional volume is distributed along the axis of the non-included intercept.

The additional uncertainty induced by the Lasso in case of no selection event can be assessed by performing inference on $\beta_2$ only.
This gives confidence intervals for the elevation coefficient 
\begin{align*}
M_{L}(\beta_2) = \left[ -14.80, -7.13\right]\times 10^{-4}, \quad
M_{WLS}(\beta_2) = \left[ -13.73,  -6.06\right]\times 10^{-4}.
\end{align*}
Both sets are constructed similarly.
While for $M_{L}$, the corresponding quantile is from the non-central $\chi_1^2$-distribution with a non-centrality parameter close to zero, $M_{WLS}$ is based upon a central $\chi_1^2$-quantile.
Therefore, the respective interval lengths are almost equal, with $M_{L}$ being narrowly wider by the sixth decimal place.
However, both intervals differ in location.
The Lasso based $M_{L}(\beta_2)$ is centered around $\hat{\beta}_{L,2}$ and thus covers a different part of the real line than the WLS based $M_{WLS}(\beta_2)$.
Consequently, although both intervals are of almost equal length, inference based on the Lasso indeed differs from WLS-based inference, even in the case of no selection event. 
The conclusions of this analysis remain valid if tuning parameters are set to the theoretical value $\lambda_1 = \lambda_2 = \sqrt{n}/2\approx 11.74$. 

Altogether, this case example confirms that if interest lies in narrow confidence sets, WLS sets are smaller and should be employed.
Simultaneously performing variable selection with the Lasso always evokes uncertainty.
Nevertheless, confidence sets based on the Lasso are not always much larger than their WLS counterparts.
When smaller models are preferable, those sets can be useful for assessing uncertainty and hypothesis testing.

\section{Discussion} \label{sec:disc}

This contribution presents a solution for inference in a
low-dimensional LMM in which the fixed effects are estimated with a
Lasso-type penalization. The uniformly valid
confidence sets for the fixed effects based on a Lasso estimator, which performs estimation and model selection simultaneously, have been constructed. We suggest a two-stage
estimation procedure, where the covariance parameters are estimated
via a REML estimator $\hat\tht$ first, and the coefficient
parameters $\bta_0$ with $\hat\tht$ plugged in second. 
In particular, this allows to treat regression coefficients and covariance
parameters separately. Using this approach, we decisively simplify the
verification of the theoretical properties of parameter estimators.
Finally, it can be shown that the resulting confidence sets are
uniformly valid over both the coefficient and covariance parameters
spaces.

It is well known that employing standard WLS-based confidence sets without
model selection will always yield smaller confidence sets than
post-selection inference \citep{Poetscher2009}. However, selecting
parsimonious models is often a key concern for practitioners. Although
resulting post-selection confidence sets will include additional
uncertainty due to variable selection, Sections~\ref{sec:sims} and
\ref{sec:app} indicate that derived confidence sets are useful. Moreover,
our results provide information about the joint estimation and variable
selection uncertainty of the Lasso in the context of LMMs by quantifying
the probability with which the estimation error $\hat\bta_L - \bta_0$
lies in the proposed ellipse.

The novelty of our method is that, to the best of our knowledge, this
work is the first that considers uniform inference based on the Lasso for an
elaborate covariance structure, as given in LMMs. Another new
important theoretical result of this article is the proof of uniform
consistency of Cram\'er type for the REML estimators of the covariance
parameters. We expect that our approach can serve as a basis for
proper inference for an estimation procedure that penalizes both fixed
and random effects.

The derived confidence regions are based on the Lasso selection
procedure, but are valid regardless of the selected model. In this,
our methodology differs from selective inference methods for LMMs
as proposed by \citep{Ruegamer2022}.

It is important to note that the proposed method does not account for
any stochastic selection of the tuning parameter $\lambda$. Instead,
the result in Theorem \ref{mainres} is a finite sample result that
holds for any a-priori choice of $\lambda$. Setting $\lambda =
n^{1/2}/2$ as in Section \ref{sec:sims} ensures conservative model
selection and uniformly valid inference. If $\lambda$ is estimated or
determined with any data-driven procedure, it emits additional
volatility that has to be quantified. This holds a fortiori for
procedures such as the adaptive Lasso where even the penalty term
itself is random. While the penalty term of the smoothly clipped
absolute deviations estimator (SCAD) \citep{Fan2001} does not exhibit
this property, it is particularly intricate to treat and leads to a
non-convex objective function. Therefore, while these methods behave
similarly to the ordinary Lasso, an extension of our findings for them
requires substantial additional research, which is out of the scope of
the current work. It should be noted that among Lasso, adaptive Lasso, and SCAD, \citep{PoetscherSchneider10} show that in the case of
orthogonal regressors, using the Lasso leads to the shortest intervals
in finite samples, strengthening our approach.

In high dimensions ($p\geq n$), confidence sets in connection with the
Lasso penalization have been addressed, e.g., by the de-biasing
approach of \citep{VdGeerEtAl14} and related works. For such settings,
the ansatz of \citep{Ewald2018} is not applicable as the matrix $\C$ is
not invertible. Moreover, the classical REML methodology is defined in
a low-dimensional setting only, as $\tht_0$ is estimated by the data
projected onto the space orthogonal to the columns of $\X$. Extensions
for REML to higher dimensions lose the property that $\hat\tht$ is
independent of $\bta_0$ \citep{Schelldorfer2011}.


\section{Proofs} \label{sec:proofs}

To prove Theorem~\ref{unifconv}, we first derive the following lemma.

\begin{lemma}  \label{Qlemma}
Let conditions (\ref{c_theta}) - (\ref{c_pd}) hold, $m =
\min\{\rk(\HH_1), \dots, \rk(\HH_r)\}$ and let $K$ be a finite set
drawn with replacement from $\{1,\dots,r\}$ with $|K| > 1$. Then, for any $\tht \in \Theta$,
\begin{equation} \label{tracesplitlemmaeq}
\tr\left\{\prod_{k \in K} \PP(\tht)\HH_k \right\} \leq
\left(\frac{c}{\underline\omega}\right)^{|K|-1} m^{1-{|K|}/{2}} \prod_{k \in K} \nu_k(\tht).
\end{equation}
\end{lemma}

The result is shown using that the eigenvalues of the matrices of
interest are smaller or equal to one and subsequently employing
regularity condition $(\ref{c_h})$.

\begin{proof}[\textbf{\upshape Proof:}]
Consider the symmetric and positive semi-definite matrices
\begin{equation}
\begin{aligned}
\A_k & = \A_k(\tht) = \V(\tht)^{-1/2}\HH_k\V(\tht)^{-1/2}, \quad k \in\{1,\dots, r\}, \\
\B = \B(\tht) &= \mathbf{I}_n - \V(\tht)^{-1/2}\X\{\X^\top\V(\tht)^{-1}\X\}^{-1}\X^\top\V(\tht)^{-1/2},
\end{aligned}
\end{equation}
so that $\tr\{\prod_{k \in K}\PP(\tht)\HH_k\} = \tr\{\prod_{k \in
K}\A_k\B\}$. Since $\B$ is a projection matrix with eigenvalues zero
or one, proceeding as in the proof of Theorem 1.1 (H\"older
inequality) from \citep{Manjegani2007} gives
\begin{align}\label{Manjegani}
\tr\left(\prod_{k \in K} \A_k\B \right)
\leq \sum_{i=1}^n \prod_{k \in K} \left\{\eta_i(\A_k)\eta_i(\B)\right\}
& \leq \sum_{i=1}^n \prod_{k \in K} \eta_i(\A_k)
\leq \prod_{k \in K} \tr\left(\A_k^{|K|}\right)^{1/|K|}.
\end{align}
We continue by evaluating $\tr(\A_k^{|K|})$. First note that
conditions (\ref{c_theta}) and (\ref{c_h}) imply
\begin{equation}\label{nubound}
c^{-1}\underline\omega \leq \theta_{k}\frac{\tr(\A_k)}{\rk(\HH_k)} \leq
c\;\overline\omega.
\end{equation}
For $k = r$, let $\eta_i' =
\eta_i(\HH_r^{-1}\sum_{l=1}^{r-1}\theta_{l}/\theta_{r}\HH_l) \geq
0$. Then, the above inequality yields
\begin{align*}
\tr\left(\A_r^{|K|}\right)
& = \frac{1}{\theta_{r}^{|K|}}\tr\left\{\left( \I_n +
\HH_r^{-1}\sum_{l=1}^{r-1}\frac{\theta_{l}}{\theta_{r}}\HH_k\right)^{-{|K|}}\right\}
= \frac{1}{\theta_{r}^{|K|}} \sum_{i=1}^n \frac{1}{(1 + \eta_i')^{|K|}}  \leq \frac{1}{\theta_{r}^{|K|}} \sum_{i=1}^n \frac{1}{1 + \eta_i'}\\
&= \frac{\tr(\A_r)}{\theta_{r}^{{|K|}-1}}
\leq \left(\frac{c}{\underline\omega}\right)^{{|K|}-1}
\frac{\tr(\A_r)^{|K|}}{\rk(\HH_r)^{{|K|}-1}},
\end{align*}
and, similarly for $k \in\{ 1, \dots, r-1\}$, let $\eta_i'' = \eta_i
\big\{\big(\sum_{l \neq k}\theta_{l}/\theta_{k}\HH_l\big)^{-1}
\HH_k\big\} \geq 0$. This gives
\begin{align*}
\tr(\A_k^{|K|})
& = \frac{1}{\theta_{k}^{|K|}}\tr\left\{\left(\I_n - \left[\I_n + \left(\sum_{l \neq k}
\frac{\theta_{l}}{\theta_{k}}\HH_l\right)^{-1}\HH_k \right]^{-1}\right)^{|K|}\right\}
= \frac{1}{\theta_{k}^{|K|}} \sum_{i=1}^n \left(\frac{\eta_i''}{1 + \eta_i''}\right)^{|K|} \\
& \leq \frac{1}{\theta_{k}^{|K|}} \sum_{i=1}^n \frac{\eta_i''}{1 + \eta_i''}
= \frac{\tr(\A_k)}{\theta_{k}^{{|K|}-1}}
\leq \left(\frac{c}{\underline\omega}\right)^{{|K|}-1}\frac{\tr(\A_k)^{|K|}}{\tr(\HH_k)^{{|K|}-1}}.
\end{align*}
Altogether, for $k\in\{1,\dots, r\}$, this gives
\begin{equation}\label{tracelemma}
\tr(\A_k^{|K|}) \leq \frac{c^{|K|-1}}{\underline\omega^{|K|-1}}\frac{\tr(\A_k)^{|K|}}{\rk(\HH_k)^{|K|-1}}.
\end{equation}
Plugging this into (\ref{Manjegani}) yields
\begin{equation*}
\prod_{k\in K} \tr\left( \A_k^{|K|}\right)^{1/|K|}
\leq \prod_{k\in K} \left(\frac{c}{\underline\omega}\right)^{\frac{|K|-1}{|K|}}
\frac{\nu_k(\tht)}{\rk(\HH_k)^{\frac{|K|-2}{2|K|}}}
\leq \left(\frac{c}{\underline\omega}\right)^{|K|-1} m^{1-|K|/2}\prod_{k \in K}\nu_k(\tht),
\end{equation*}
which gives the claim.
\end{proof}
Note that under the conditions of Lemma \ref{Qlemma} but with $|K| =
1$, (\ref{Manjegani}) implies that
\begin{equation}\label{Kis1}
\tr\left\{\PP(\tht)\HH_k\right\} \leq \tr\left\{\A_k(\tht)\right\}.
\end{equation}

Before getting into the proof of Lemma~\ref{unifconv}, we denote from
now on $\Q_i(\tht) = \PP(\tht)\HH_i$, $\Q_{ij}(\tht) =
\prod_{l\in\{i,j\}} \PP(\tht)\HH_l$, $\Q_{ijk}(\tht) =
\prod_{l\in\{i,j,k\}} \PP(\tht)\HH_l$ to facilitate notation. In
particular, Lemma~\ref{Qlemma} implies that
\begin{equation} \label{trsplitREML}
  -\mbox{E}_{\tht_0}\left\{
  \frac{\partial^2\ell_R}{\partial\theta_i^2} (\tht)
  \bigg|_{\tht_0}\right\}
  = \frac{1}{2}\tr\left\{\Q_{ii}(\tht_0) \right\}
  = O\{\nu_i(\tht_0)^2\}.
\end{equation}

It is straightforward to adapt Lemma~\ref{Qlemma} to the case in which
$\PP(\tht)$ is replaced by $\V(\tht)^{-1}$:

\begin{cor} \label{corollary}
Let model (\ref{GenLM}) and (\ref{c_theta}) - (\ref{c_pd}) hold and
let $K$ be a set drawn with replacement from $\{1, \dots,  r\}$ with $|K|>1$. Then, for $\tht\in\Theta$, 
\begin{equation}
\tr\left\{\prod_{k \in K} \V(\tht)^{-1}\HH_i \right\} \leq
\left(\frac{c}{\underline\omega}\right)^{|K|-1} \prod_{k \in K}
\frac{\nu_i(\tht)}{\rk(\HH_i)^{\frac{|K|-2}{2|K|}}}.
\end{equation}
\end{cor}

Above results are related to the reverse inequality from \cite[p. 232, Problem 4.1.P13]{Horn2013}, namely that for a positive definite matrix $\Z\in\mathbb{R}^{n\times n}$, $n\in\mathbb{N}$, it holds
\begin{equation} \label{HornJohnson}
	\tr(\Z^2) \geq \frac{\tr(\Z)^2}{\rk{(\Z)}}.
\end{equation}

%

%
\begin{proof}[\textbf{\upshape Proof of Theorem~\ref{unifconv}:}]
%
This proof is adapted to account for uniformity w.r.t. $\tht_0$ from Lemma 7.2 (i) from \citep{Jiang1996} and closely follows the lines of \citep{Weiss1971}.
%
The four boundedness properties below on the derivatives of the log-likelihood are surrogates of the boundedness conditions imposed on the log-likelihood directly by \citep{Moran1970, Wald1949}.
By a Taylor expansion, the conditions below also include the boundedness of the log-likelihood.
They are further required to omit the compactness condition of $\Theta$, by constructing a compact set $\Theta^{\prime}_{\epsilon}$, see details below.
For all $i,j\in\{1,\dots, r\}$, it holds for any $\tht_0$ that
\begin{enumerate}[(i)]
\item $\mbox{E}_{\tht_0}\left\{ \displaystyle \frac{\partial \ell_R}{\partial \theta_i} (\tht)\bigg|_{\tht_0}  \right\} = 0$,
\item $\displaystyle \frac{1}{\nu_i(\tht_0)} \frac{\partial \ell_R}{\partial \theta_i} (\tht)\bigg|_{\tht_0} = O_P(1)$,
\item 
For $\J(\tht_0)$ with
$\displaystyle J_{ij}(\tht_0) = -\frac{1}{\nu_i(\tht_0)\nu_j(\tht_0)}\mbox{E}_{\tht_0}\left\{\frac{\partial^2 \ell_R}{\partial\theta_i\theta_j}(\tht)\bigg|_{\tht_0}\right\}$ it holds that
$\eta_r \{\J(\tht_0)\} \geq c' > 0$, for some constant $c'>0$.
\item $\zeta(\tht, \tht_0) = \displaystyle \frac{1}{\nu_i(\tht_0)\nu_j(\tht_0)} \frac{\partial^2 \ell_R}{\partial \theta_i\partial \theta_j} (\tht) + J_{ij}(\tht_0) = O_P\left(\frac{q}{\sqrt{m}}\right)$ for $\tht\in\Theta_{q} =  \{ \tht: \nu_i(\tht_0)|\theta_i - \theta_{0,i}| \leq q \}$,
\end{enumerate}
where $q= m^{1/(6+\delta)}$ for some $\delta>0$ and $m = \min_{k\in\{1,\dots, r\}}\rk(\HH_i)$.
For readability, suppress the dependency from $\tht$ when the argument is clear from the context.
Now, (i)-(iv) are shown.
\begin{enumerate}[(i)]
%
\item Note that ${\partial\PP}/{\partial\theta_i} = -\PP\HH_i\PP$ and $\PP\V\PP = \PP$, as well as
\begin{align*}
\frac{\partial \ell_R}{\partial \theta_i}\bigg|_{\tht_0}
	&=  {\frac{1}{2}}\y^\top\PP\HH_i\PP\y -{\frac{1}{2}}\tr\left( \PP{\HH_i} \right).
\end{align*}
The claim now follows after taking expectations.
%
\item As $\text{E}( {\partial \ell_R}/{\partial \theta_i}|_{\tht_0} )=0$ by (i),
 Chebyshev's inequality can be applied, it holds uniformly and gives that for any $\epsilon>0$ there exists $k>0$,  such that
\begin{align*}
	\sup_{\tht_0\in\ThetaC}\Prob_{\tht_0}\left\{\text{Var}_{\tht_0}\left( \frac{\partial \ell_R}{\partial \theta_i}\bigg|_{\tht_0} \right)^{-1/2} \left| \frac{\partial \ell_R}{\partial \theta_i}\bigg|_{\tht_0} \right| \geq k \right\} \leq \frac{1}{k^2} < \epsilon,
\end{align*}
and thus
\begin{align*}
\frac{\partial \ell_R}{\partial \theta_i}\bigg|_{\tht_0}
= O_P\left\{ \text{Var}_{\tht_0}\left( \frac{\partial \ell_R}{\partial \theta_i}\bigg|_{\tht_0} \right)^{1/2} \right\}
= O_P(\nu_i),
\end{align*}
where the last equation follows by $\text{Var}_{\tht_0}( {\partial \ell_R}/{\partial \theta_i}|_{\tht_0} ) = -\mbox{E}_{\tht_0}\{\partial^2\ell_R/\partial\theta_i^2 |_{\tht_0}\}$
and (\ref{trsplitREML}).
%
\item By (\ref{nubound}), for all $i,j \in\{1,\dots, r\}$,
\begin{align*}
J_{ij} 
= \frac{\tr(\Q_{ij})}{ 2\nu_i\nu_j }
\geq
\frac{\theta_{0,i}\theta_{0,j}\tr(\Q_{ij})}{2c^2\overline\omega^2\sqrt{\rk(\HH_i)\rk(\HH_j)}}.
\end{align*}
By a similar argument,
\begin{align*}
	\theta_{0,i}\theta_{0,j}\tr(\Q_{ij})
	&= \theta_{0,i} \theta_{0,j} \tr\big[ \V(\tht_0)^{-1}[\I_n - \X\{\X^\top\V(\tht_0)^{-1}\X\}^{-1}\X^\top\V(\tht_0)^{-1}]\HH_i\PP(\tht_0)\HH_j \big]\\
	&\geq \theta_{0,j}\frac{\min(\tht)}{\max(\tht)} \tr\big[ \V(\mathbf{1}_r)^{-1}[\I_n - \X\{\X^\top\V(\tht_0)^{-1}\X\}^{-1}\X^\top\V(\tht_0)^{-1}]\HH_i\PP(\tht_0)\HH_j \big] \\
	&\geq \frac{\theta_{0,j}}{c} \tr\big[ \V(\mathbf{1}_r)^{-1}[\I_n - c\X\{\X^\top\V(\mathbf{1}_r)^{-1}\X\}^{-1}\X^\top\V(\mathbf{1}_r)^{-1}]\HH_i\PP(\tht_0)\HH_j \big] \\
	&\geq \frac{1}{c^2} \tr\{ \PP(\mathbf{1}_r, c)\HH_i\PP(\mathbf{1}_r, c)\HH_j \}
\end{align*}
It follows that
\begin{equation*}
	J_{ij}(\tht_0) \geq \frac{ \tr\{ \PP(\mathbf{1}_r, c)\HH_i\PP(\mathbf{1}_r, c)\HH_j \} }{2c^4\overline\omega^2\sqrt{\rk(\HH_i)\rk(\HH_j)}}.
\end{equation*}
Now, for $\aaa\in\mathbb{R}^r$ with $\aaa^\top\aaa = 1$ eigenvector to $\eta_r\{\J(\tht_0)\}$,
\begin{align*}
	\eta_r\{\J(\tht_0)\} = \aaa^\top\J(\tht_0)\aaa = \sum_{i=1}^r\sum_{j=1}^r a_i a_j J_{ij}(\tht_0) \geq \frac{\eta_r\{\KK\}}{2c^4\overline\omega^2} \geq \frac{b}{2c^4\overline\omega^2} > 0.
\end{align*}

%
\item First, note that for $\tht\in\Theta_q$, (\ref{nubound}) gives 
\begin{equation*}
	\frac{\underline\omega\sqrt{\rk(\HH_i)}}{c\theta_i}
	\leq \nu_i(\tht) \leq 
	\frac{c\overline\omega\sqrt{\rk(\HH_i)}}{\theta_i}.
\end{equation*}
For any $t\in[0,1]$ and $\tilde\tht = \tht_0 + t(\tht - \tht_0)$ this implies with (\ref{c_h}) and Corollary \ref{corollary} that 
\begin{align*}
	\tr\left\{ \V(\tilde\tht)^{-1}\HH_j\V(\tilde\tht)^{-1}\HH_i \right\}
	&\leq \frac{c}{\theta_{0,j}\theta_{0,i}} \tr\left\{ \V(\mathbf{1}_r)^{-1}\HH_j\V(\mathbf{1}_r)^{-1}\HH_i \right\} \\&\leq \frac{c^2}{\underline\omega} \frac{\tr\left\{ \V(\mathbf{1}_r)^{-1}\HH_j\right\} \tr\left\{\V(\mathbf{1}_r)^{-1}\HH_i \right\}}{\theta_{0,j}\theta_{0,i}\sqrt{\rk(\HH_j)\rk(\HH_i)}}\\
	&\leq \frac{c^2\overline\omega^2}{\underline\omega}\frac{\sqrt{\rk(\HH_j)\rk(\HH_i)}}{\theta_{0,j}\theta_{0,i}}\leq c^3\frac{\overline\omega^2}{\underline\omega^2}\nu_i(\tht)\nu_j(\tht).
\end{align*}
A Taylor expansion around $\tht_0$ yields
\begin{equation}
\begin{aligned}\label{nutheta}
\nu_i(\tht)
  &= \nu_i(\tht_0) + \frac{1}{\sqrt{\rk(\HH_i)}}\sum_{j=1}^r (\theta_j - \theta_{0,j}) \tr\left\{ \V(\tilde\tht)^{-1}\HH_j\V(\tilde\tht)^{-1}\HH_i \right\} \\&\leq \nu_i(\tht_0)\left\{1 + \frac{\overline\omega^2c^3}{\underline\omega^2\sqrt{m}}\sum_{j=1}^r \nu_j(\tht)(\theta_j - \theta_{0,j})\right\}\\
  &\leq \nu_i(\tht_0)\left\{1 + \frac{r\overline\omega^2c^3q}{\underline\omega^2\sqrt{m}}\right\}.
\end{aligned}
\end{equation}
Below, abbreviate the constant $c' = {r\overline\omega^2c^3}/{\underline\omega^2}$ and note that $q / \sqrt{m} \rightarrow 0$. 
This implies that
\begin{align} \label{thtscale}
  \tht \in \Theta_q \quad \Rightarrow \quad
  \nu_i(\tht)|\theta_i - \theta_{0,i}| \leq q\left(1 + c'\frac{q}{\sqrt{m}}\right).
\end{align}
Now, by (\ref{thtscale}), it follows for $\tht\in\Theta_q$:  
\begin{equation*}
	\V(\tht_0) = \V(\tht) + \sum_{k=1}^r (\theta_{0,k} - \theta_k)\HH_k
	\leq \V(\tht) + \sum_{k=1}^r \HH_k \frac{q}{\nu_k(\tht)}\left(1 + c'\frac{q}{\sqrt{m}}  \right).
\end{equation*}
With Lemma~\ref{Qlemma}, this gives
\begin{align*}
\mbox{E}_{\tht_0}\left\{-\frac{1}{\nu_i(\tht_0)\nu_j(\tht_0)} \frac{\partial^2\ell_R}{\partial\theta_i\partial\theta_j}(\tht) \right\}
	&= -\frac{\tr\{\Q_{ij}(\tht)\}/2 - \tr\left\{\Q_{ij}(\tht)\PP(\tht)\V(\tht_0)\right\}}{\nu_i(\tht_0)\nu_j(\tht_0)}  \\
	&\leq \frac{\tr\{\Q_{ij}(\tht)\}}{2\nu_i(\tht_0)\nu_j(\tht_0)} + \sum_{k=1}^r q\frac{\tr\{\Q_{ijk}(\tht)\}}{\nu_i(\tht)\nu_j(\tht)\nu_k(\tht)} \left(1 + c'\frac{q}{\sqrt{m}}  \right)^3 \\&= \frac{\tr\{\Q_{ij}(\tht)\}}{2\nu_i(\tht_0)\nu_j(\tht_0)}  + O\left(\frac{q}{\sqrt{m}}  \right),
\end{align*}
for $m = \min_{k\in\{1,\dots, r\}}\rk(\HH_i)$ and all  $\tht\in\Theta_{q}$.
Similarly, now dropping the argument for all quantities depending on $\tht\in\Theta_q$ for the sake of readability, i.e. $\nu_i = \nu_i(\tht)$, 
\begin{align*}
\mbox{Var}_{\tht_0}\left\{-\frac{1}{\nu_i(\tht_0)\nu_j(\tht_0)} \frac{\partial^2\ell_R}{\partial\theta_i\partial\theta_j} \right\}
	&= 2\frac{\tr\left\{ \Q_{ij}\PP\V(\tht_0)\Q_{ij}\PP\V(\tht_0) \right\}}{\nu_i(\tht_0)^2\nu_j(\tht_0)^2} 
	+ O\left(\frac{q^2}{m}\right) \\
	&\hspace{-3.5cm}= O\left\{ \frac{\tr(\Q_{ij}\Q_{ij})}{\nu_i^2\nu_j^2}
					+ 2q\sum_{k=1}^r\frac{\tr(\Q_{ij}\Q_{ijk})}{\nu_i^2\nu_j^2\nu_k}
					+ q^2\sum_{k,l=1}^r\frac{\tr(\Q_{ijk}\Q_{ijl})}{\nu_i^2\nu_j^2\nu_l\nu_k} + \frac{q^2}{m}\right\}\\&= O\left( \frac{1}{m^{1/2}} + \frac{2q}{m^{3/2}} + \frac{q^2}{m^2} + \frac{q^2}{m}\right) = O\left(\frac{q^2}{m}\right).
\end{align*}
Putting the previous two results together and proceeding as in (ii) with Chebyshev's inequality, this gives that for any $\tht_0$ it holds
\begin{equation*}
  -\frac{1}{\nu_i(\tht_0)\nu_j(\tht_0)}   \frac{\partial^2\ell_R}{\partial\theta_i\partial\theta_j}(\tht) = \frac{\tr\{\Q_{ij}(\tht)\}}{2\nu_i(\tht_0)\nu_j(\tht_0)} + O_P\left(\frac{q}{\sqrt{m}}\right)
\end{equation*}
for $\tht\in\Theta_{q}$.
To prove (iv), this must hold for $\tht_0$ on the right-hand side. 
Again, let $t\in[0,1]$ and $\tilde\tht = \tht_0 + t(\tht - \tht_0)$. A Taylor expansion of $\tr\{\Q_{ij}(\tht)\}$ around $\tr\{\Q_{ij}(\tht_0)\}$ gives 
\begin{align*}
	\frac{\tr\{\Q_{ij}(\tht)\}}{2\nu_i(\tht_0)\nu_j(\tht_0)} &= \frac{\tr\{\Q_{ij}(\tht_0)\}}{2\nu_i(\tht_0)\nu_j(\tht_0)} + \sum_{k=1}^r (\theta_{0,k}-\theta_k) \frac{\tr\{\Q_{ijk}(\tilde\tht) + \Q_{ikj}(\tilde\tht)\}}{\nu_i(\tht_0)\nu_j(\tht_0)}\\&\hspace*{-2cm}\leq J_{ij}(\tht_0) + q \frac{\tr\{\Q_{ijk}(\tilde\tht) + \Q_{ikj}(\tilde\tht)\}}{\nu_k(\tht_0)\nu_i(\tht_0)\nu_j(\tht_0)} \\
	&\hspace*{-2cm}\leq J_{ij}(\tht_0) + q \frac{\tr\{\Q_{ijk}(\tilde\tht) + \Q_{ikj}(\tilde\tht)\}}{\nu_k(\tilde\tht)\nu_i(\tilde\tht)\nu_j(\tilde\tht)}\left(1 + c'\frac{q}{\sqrt{m}}  \right)^3= J_{ij}(\tht_0) + O\left( \frac{q}{\sqrt{m}} \right)
\end{align*}
Altogether, 
\begin{equation*}
-\frac{1}{\nu_i(\tht_0)\nu_j(\tht_0)}\frac{\partial^2 \ell_R}{\partial \theta_i\partial \theta_j}(\tht)
	= J_{ij}(\tht_0) + O_P\left(\frac{q}{\sqrt{m}}\right).
\end{equation*}
\end{enumerate}

The second part of the proof mimics the reasoning of \citep{Weiss1971}. Let $\J(\tht_0) = \{J_{ij}(\tht_0)\}_{ij}$ and
$\mathbf{s}(\tht_0) = \{s_1(\tht_0),s_2(\tht_0),\dots,s_r(\tht_0)\}^\top$ with $s_i(\tht_0) = \nu_i(\tht_0)^{-1} \partial \ell_R/\partial \theta_i |_{\tht_0}$ and define $\tht^{\prime}$ such that
\begin{align*}
\nu(\tht_0) \circ (\tht^{\prime} - \tht_0) =  \J(\tht_0)^{-1}\mathbf{s}(\tht_0).
\end{align*}
We will show that there exists a $\hat\tht$ so that
\begin{align*}
	\nu_i(\tht_0)|\hat\theta_i - \theta_{0,i}|
	&\leq \nu_i(\tht_0)|\hat\theta_i - \theta_{i}'| + \nu_i(\tht_0)|\theta_{i}' - \theta_{0,i}|
\end{align*}
By (ii) and (iii), the last term is $\nu_i(\tht_0)|\theta_{i}' - \theta_{0,i}| =O_P(1)$. It remains to show that $\nu_i(\tht_0)|\hat\theta_i - \theta_{i}'| = O_P(1)$.
This is shown by proving that a maximum of $(\ref{REML})$ is attained in a region close to $\tht'$. 

Now, for any $\tht\in\Theta_q$,
\begin{equation}
\begin{aligned} \label{Taylor_for_theta}
\ell_R(\tht) &= \ell_R(\tht_0)
					+ \sum_{i=1}^r \nu_i(\tht_0)(\theta_i - \theta_{0,i})s_i(\tht_0) \\
					&\hspace{40pt}- \frac{1}{2}\sum_{i=1}^r\sum_{j=1}^r \nu_i(\tht_0)(\theta_i - \theta_{0,i}) \nu_j(\tht_0)(\theta_j - \theta_{0,j}) J_{ij}(\tht_0) + R(\tht,\tilde\tht)\\
   			 &= \ell_R(\tht_0)
					+ \frac{1}{2}\mathbf{s}(\tht_0)^\top \J(\tht_0)^{-1} \mathbf{s}(\tht_0) \\
					&\hspace{40pt}- \frac{1}{2}\sum_{i=1}^r\sum_{j=1}^r \nu_i(\tht_0)(\theta_i^{\prime} - \theta_{i}) \nu_j(\tht_0)(\theta_j^{\prime} - \theta_{j}) J_{ij}(\tht_0)
					 + R(\tht,\tilde\tht),
\end{aligned}
\end{equation}
where $R(\tht,\tilde\tht) = -\frac{1}{2}\sum_{i=1}^r\sum_{j=1}^r  \nu_i(\tht_0)(\theta_i - \theta_{0,i}) \nu_j(\tht_0)(\theta_j - \theta_{0,j})\zeta(\tilde\tht, \tht_0)$ for some $\tilde\theta_i = \theta_{0,i} + t(\theta_{i} - \theta_{0,i})$ where $t\in[0,1]$.
Now consider the set $\Theta_{\epsilon_n}^{\prime} = \big\{\tht: \nu_i(\tht_0)|\theta_i-\theta_i^{\prime}| < \epsilon_n \big\}$ and its boundary $\partial{\Theta}_{\epsilon_n}^{\prime}$ and $\overline{\Theta}_{\epsilon_n}^{\prime} = {\Theta}_{\epsilon_n}^{\prime} \cup \partial{\Theta}_{\epsilon_n}^{\prime}$. 
For any $\tht\in\overline{\Theta}_{\epsilon_n}^{\prime}$ with $\epsilon_n\rightarrow0$, 
\begin{equation*}
	\nu_i(\tht_0)|\theta_i - \theta_{0,i}|
	\leq \nu_i(\tht_0)|\theta_i - \theta_{i}'|  + \nu_i(\tht_0)|\theta_i' - \theta_{0,i}|
	= O_P(1),
\end{equation*}
that is for all $\varepsilon>0$ there exists a constant $C$ and $N_\varepsilon$ such that for all $n\geq N_\varepsilon$: 
$1 - \varepsilon \leq \inf_{\tht_0\in\Theta} \Prob_{\tht_0}\{\nu_i(\tht_0)|{\theta}_{i} - \theta'_i|< C\}$. 
Since $q\rightarrow\infty$, there exists $N\geq N_\varepsilon$ such that $q>C$. 
Together, for any $\tht\in\overline{\Theta}_{\epsilon_n}^{\prime}$ and any 
$\varepsilon>0$ there exists $N\in\mathbb{N}$, such that  for all $n\geq N$:
\begin{align*}
	1 - \varepsilon 
	\leq \inf_{\tht_0\in\Theta} \Prob_{\tht_0}\left\{\nu_i(\tht_0)|{\theta}_{i} - \theta'_i|< C\right\}
	\leq \inf_{\tht_0\in\Theta} \Prob_{\tht_0}\left\{\nu_i(\tht_0)|{\theta}_{i} - \theta'_i|< q\right\},
\end{align*}
i.e., $\inf_{\tht_0\in\Theta} \Prob_{\tht_0}( \overline\Theta'_{\epsilon_n} \subset \Theta_q ) \rightarrow 1$. 
This implies that 
\begin{align*}
	&\inf_{\tht_0\in\Theta} \Prob_{\tht_0}\left\{
\sup_{\tht\in\partial{\Theta}_{\epsilon_n}^{\prime}} | R(\tht,\tilde\tht)|
\leq 
\sup_{\tht\in\Theta_q} | R(\tht, \tilde\tht)|  \right\} 
\\\geq\phantom{+}\;&\inf_{\tht_0\in\Theta} \Prob_{\tht_0}\left\{
\sup_{\tht\in\partial{\Theta}_{\epsilon_n}^{\prime}} | R(\tht,\tilde\tht)|
\leq
\sup_{\tht\in\Theta_q} | R(\tht, \tilde\tht)| \;\bigg|\; \partial\Theta'_{\epsilon_n} \subset \Theta_q\right\}\inf_{\tht_0\in\Theta} \Prob_{\tht_0}\left( \partial\Theta'_{\epsilon_n} \subset \Theta_q \right)\\ 
+&\inf_{\tht_0\in\Theta} \Prob_{\tht_0}\left\{
\sup_{\tht\in\partial{\Theta}_{\epsilon_n}^{\prime}} | R(\tht,\tilde\tht)|
\leq
\sup_{\tht\in\Theta_q} | R(\tht, \tilde\tht)| \;\bigg|\; \partial\Theta'_{\epsilon_n} \not\subset \Theta_q\right\}\inf_{\tht_0\in\Theta} \Prob_{\tht_0}\left( \partial\Theta'_{\epsilon_n} \not\subset \Theta_q \right) 
\rightarrow 1.
\end{align*}
Second-to-last, consider
\begin{align*}
\delta(\epsilon_n) = \min_{\tht\in\partial{\Theta}_{\epsilon_n}^{\prime}}
	\frac{1}{2}\sum_{i=1}^r\sum_{j=1}^r \nu_i(\tht_0)(\theta_i^{\prime} - \theta_{i}) \nu_j(\tht_0)(\theta_j^{\prime} - \theta_{j}) J_{ij}(\tht_0) = \frac{\epsilon_n^2}{2}\mathbf{1}_r^\top\mathbf{J}(\tht_0)\mathbf{1}_r \geq c'\frac{r\epsilon_n^2}{2} > 0
\end{align*}
by (iii). Note that $\delta(\epsilon_n)$ is not stochastic, increasing in $\epsilon_n$ and $\delta(0)=0$. 
Further, by (iv),
$
2\sup_{\tht\in\Theta_q} | R(\tht, \tilde\tht)| = O_P(q^3/\sqrt{m})
$, 
i.e., for all $\varepsilon>0$ there exists a constant $C'$ and $N'_\varepsilon$ such that for all $n\geq N'_\varepsilon$ it holds that
	$1 - \varepsilon \leq \inf_{\tht_0\in\Theta} \Prob_{\tht_0}\{2\sup_{\tht\in\Theta_q} | R(\tht, \tilde\tht)|< C' {q^3}/{\sqrt{m}}\}$. 
Now set $\epsilon_n\rightarrow0$ so that it exists $N'>N'_{\varepsilon}$ such that for all $n>N'\colon\epsilon_n^2rc'/2 > C'q^3/\sqrt{m}$.
This gives that for all $n>N'$: 
\begin{align*}
	1 - \varepsilon &\leq \inf_{\tht_0\in\Theta} \Prob_{\tht_0}\left\{2\sup_{\tht\in\Theta_q} | R(\tht, \tilde\tht)|< C' \frac{q^3}{\sqrt{m}}\right\} \\&\leq \inf_{\tht_0\in\Theta} \Prob_{\tht_0}\left\{2\sup_{\tht\in\Theta_q} | R(\tht, \tilde\tht)|< c'\frac{r\epsilon_n^2}{2} \right\} \leq \inf_{\tht_0\in\Theta} \Prob_{\tht_0}\left\{2\sup_{\tht\in\Theta_q} | R(\tht, \tilde\tht)|< \delta(\epsilon_n)\right\}.
\end{align*}
Finally, this implies that for $\tht'\in\Theta_q$ and by (\ref{Taylor_for_theta})
\begin{align*}
\lim_{n\rightarrow\infty} \inf_{\tht_0\in\ThetaC} \Prob_{\tht_0}\left\{   \ell_R(\tht')  >
\max_{\tht\in\partial{\Theta}_{\epsilon_n}^{\prime}}
 \ell_R(\tht) \right\} &= \lim_{n\rightarrow\infty} \inf_{\tht_0\in\ThetaC}
		\Prob_{\tht_0}\left\{ \delta(\epsilon_n) > \max_{\tht\in\partial{\Theta}_{\epsilon_n}^{\prime}} R(\tht,\tilde\tht) - R(\tht',\tilde{\tht}^{\prime}) \right\}\\
	&\geq \lim_{n\rightarrow\infty} \inf_{\tht_0\in\ThetaC}
		\Prob_{\tht_0}\left\{ \delta(\epsilon_n) > 2\sup_{\tht\in\Theta_q} | R(\tht, \tilde\tht)| \right\} = 1,
\end{align*}
i.e., with probability going to one, $\ell_R$ attains a local maximum in $\Theta'_{\epsilon_n}$. 
Denote $\hat\tht$ as a local maximum of $\ell_R$ in $\Theta'_{\epsilon_n}$ if one exists and set $\hat\tht\not\in\Theta'_{\epsilon_n}$ otherwise.  
Then, above statement gives that 
for all $\varepsilon>0$ there exists $N^{\prime\prime}_\varepsilon$ such that for all $n\geq N^{\prime\prime}_\varepsilon$: $\inf_{\tht_0\in\ThetaC}\Prob_{\tht_0}(\hat\tht\in\Theta'_{\epsilon_n})\geq 1-\varepsilon$. 
Now, for any constant $C^{\prime\prime}>0$ and any $\varepsilon>0$ let 
$N^{\prime\prime}$ be such that $C^{\prime\prime}>\epsilon_{N^{\prime\prime}}$ and $N^{\prime\prime}>N^{\prime\prime}_\varepsilon$. 
Then, for all $n>N^{\prime\prime}$, 
\begin{align*}
	\sup_{\tht_0\in\Theta}\Prob_{\tht_0}\left\{\nu_i(\tht_0)|\hat{\theta}_{i} - \theta'_i|\geq C^{\prime\prime}\right\} \leq \phantom{+}\;&\sup_{\tht_0\in\Theta}\Prob_{\tht_0}\left\{\nu_i(\tht_0)|\hat{\theta}_{i} - \theta'_i|\geq C^{\prime\prime}, \hat\tht\in\Theta'_{\epsilon_n} \right\} \\+ 
	   &\sup_{\tht_0\in\Theta}\Prob_{\tht_0}\left\{\nu_i(\tht_0)|\hat{\theta}_{i} - \theta'_i|\geq C^{\prime\prime}, \hat\tht\not\in\Theta'_{\epsilon_n} \right\} \\
	\leq &\sup_{\tht_0\in\Theta}\Prob_{\tht_0}\left\{\hat\tht\not\in\Theta'_{\epsilon_n} \right\} \leq \epsilon, 
\end{align*}
that is $\nu_i(\tht_0)|\hat{\theta}_{0,i} - \theta'_i| = O_P(1)$. 
\end{proof}

%
To address the infimum over $\Theta$, we use the following result.
\begin{lemma} \label{pullerroutlemma}
For $n\in\mathbb{N}$, let $(X_n)$ and $(Y_n)$ be sequences of real-valued continuous random variables with finite variances, in $x$ and $y$ uniformly bounded conditional density $p_{X_n|Y_n=y}(x)=O(1)$
 and $X_n = O_P(a_n)$ and $Y_n=O_P(a_nb_n)$ with sequences $(a_n), (b_n)\in\mathbb{R}$ with $b_n = o(1)$. Then, as $n \to \infty$,
\begin{align*}
\Prob\big(X_n + Y_n \leq a_n\big) &= \Prob(X_n\leq a_n) + O(b_n).
\end{align*}
\end{lemma}
The asymptotic result is clear as convergence in probability implies convergence in distribution.
The above result further specifies the rate of convergence.

\begin{proof}[\textbf{\upshape Proof:}]
First, let $\phi(s,t) = P(X_n + s \leq a_n | Y_n = t )$ and consider a Taylor expansion for 
$f(x) = P(X_n / a_n + x \leq 1 | Y_n = t )$ around $x = 0$, which gives $\phi(s,t) = \phi(0,t) + O(s/a_n) $.
This implies $\int_{a_n-t}^{a_n} p_{X_n|Y_n = t}(z) \;dz = \phi(0,t) - \phi(t,t) = O(t/a_n) $.
Using convolution the formula we obtain
\begin{equation}
\begin{aligned}
\Prob\big(X_n + Y_n \leq a_n\big)
	&= \int_{-\infty}^{a_n} p_{X_n+Y_n}(z) \;dz = \int_{-\infty}^{a_n} \int_{-\infty}^{\infty} p_{X_n,Y_n}(z-t,t) \;dt \;dz\\&= \int_{-\infty}^{a_n} \int_{-\infty}^{\infty} p_{X_n|Y_n = t}(z-t) p_{Y_n}(t) \;dt\;dz\\
	&= \int_{-\infty}^{\infty} p_{Y_n}(t) \int_{-\infty}^{a_n-t} p_{X_n|Y_n = t}(z) \;dz\;dt \\&= \int_{-\infty}^{\infty} p_{Y_n}(t) \bigg\{\int_{-\infty}^{a_n} p_{X_n|Y_n = t}(z) \;dz+ \int_{a_n-t}^{a_n} p_{X_n|Y_n = t}(z) \;dz\bigg\} \;dt \\
	&= \int_{-\infty}^{\infty} p_{Y_n}(t) \bigg\{\int_{-\infty}^{a_n} p_{X_n|Y_n = t}(z) \;dz+ O\left(\frac{t}{a_n}\right)\bigg\} \;dt\\&= \int_{-\infty}^{\infty} \int_{-\infty}^{a_n} p_{X_n,Y_n}(z,t)  \;dz \;dt + O\left\{\mbox{E}\left( \frac{Y_n}{a_n}\right) \right\}\\
	&= \int_{-\infty}^{\infty} \int_{-\infty}^{a_n} p_{X_n}(z)p_{Y_n|X_n=z}(t)  \;dz \;dt + O(b_n)= \Prob(X_n\leq a_n) + O(b_n),
\end{aligned}
\end{equation}
which gives the claim.
\end{proof}
The next two results are helpful to prove Theorem~\ref{mainres}.
\begin{lemma} \label{uniflimpiv}
Let model (\ref{GenLM}) and (\ref{c_theta})-(\ref{c_pd}) hold and set $m=\min\{ \rk(\HH_1), \dots, \rk(\HH_r)\}\rightarrow\infty$ and $\hat t = t + \omega >0$ with
\begin{align*}
t = O\left(p+ \left\| \C^{-1/2}\LAM \dd \right\|^2 \right), \qquad \omega =  O_P\left( \frac{1}{\sqrt{m}}\left\| \C^{-1/2}\LAM \dd \right\|^2 \right).
\end{align*}
Then, there exists a sequence of local maximizers $\hat\tht$ of (\ref{REML})
such that
\begin{align*}
\inf_{\substack{\dd\in\{-1,1\}^p\\\tht_0\in\ThetaC}}\Prob_{\tht_0}\left\{ \hat\uu_\dd \in E\left(\hat\C, \hat t\right) \right\}
	&= \inf_{\substack{\dd\in\{-1,1\}^p\\\tht_0\in\ThetaC}}\Prob_{\tht_0}\left\{ \uu_\dd \in E\left(\C,t\right) \right\} + O\left(\frac{1}{\sqrt{m}}\right)
\end{align*}
uniformly over $\tht_0\in\Theta$. 
\end{lemma}
%
\begin{proof}[\textbf{\upshape Proof:}]
First, let $\xi = \|\C^{-1/2}\LAM\dd\|^2$ and observe that $\|\C^{1/2}\uu_\dd\|^2 \sim \chi^2_p( \xi )$.
This implies that $\|\C^{1/2}\uu_\dd\|^2 = O_P(1+\xi)$ for all $\tht_0$.
Now, consider its derivative with respect to $\theta_i$, $i\in\{1,\dots, r\}$.
\begin{align*}
\frac{\partial}{\partial\theta_i} \|\C^{1/2}\uu_\dd\|^2
	&= \|\OMEGA_i^{1/2}(\w-\LAM\dd)\|^2 + 2(\w-\LAM\dd)^\top\C^{-1}\frac{\partial\w}{\partial\theta_i},
\end{align*}
where $\OMEGA_i = \C^{-1}\X^\top\V^{-1}{\HH_i}\V^{-1}\X\C^{-1}/n$ as $\partial\C^{-1}/\partial\theta_i = \OMEGA_i$. Now,
\begin{align*}
\mbox{E}\left(\frac{\partial}{\partial\theta_i} \|\C^{1/2}\uu_\dd\|^2 \right)
	&= \|\OMEGA_i^{1/2}\LAM\dd\|^2 - \tr(\C\OMEGA_i)
	\leq \|\OMEGA_i^{1/2}\LAM\dd\|^2 + \tr(\C\OMEGA_i).
\end{align*}
Let $\A_i = \V^{-1/2}\HH_i\V^{-1/2}$ and $\B = \V^{-1/2}\X\C^{-1}(\LAM\dd\dd^\top\LAM + \C)\C^{-1}\X^\top\V^{-1/2}/n$. 
Both $\A_i$ and $\B$ are positive semi-definite. 
Since $\rk(\B) = p$ and $\eta_1(\B) = O(1 + \xi)$, 
\begin{align*}
\mbox{E}\left(\frac{\partial}{\partial\theta_i} \|\C^{1/2}\uu_\dd\|^2 \right)
	&\leq \tr(\A_i\B)
	\leq \sum_{j=1}^p \eta_j(\A_i)\eta_j(\B) 
	 = O\left\{ (1+\xi)\sum_{j=1}^p\eta_j(\A_i)\right\}
\end{align*}
for all $\tht_0$.
Similarly, lengthy calculations give that
\begin{align*}
\mbox{Var}\left(\frac{\partial}{\partial\theta_i} \|\C^{1/2}\uu_\dd\|^2 \right)
	&= 2\tr(\OMEGA_i\C\OMEGA_i\C)
	- 4\left\|\C^{1/2}\OMEGA_i\LAM\dd\right\|^2
	+\frac{4}{{n}}\left\|\V^{-1/2}\HH_i\V^{-1}\X\C^{-1}\LAM\dd\right\|^2\\
	&\leq 2\tr(\OMEGA_i\C\OMEGA_i\C)
	+ 4\left\|\C^{1/2}\OMEGA_i\LAM\dd\right\|^2
	+\frac{4}{{n}}\left\|\V^{-1/2}\HH_i\V^{-1}\X\C^{-1}\LAM\dd\right\|^2.
\end{align*}
Proceeding as above, for all $\tht_0$, 
\begin{align*}
\mbox{Var}\left(\frac{\partial}{\partial\theta_i} \|\C^{1/2}\uu_\dd\|^2 \right)
	&= O\left[ (1+\xi)\left\{\sum_{j=1}^p\eta_j(\A_i)\right\}^2\right].
\end{align*}
Noting that $c\eta_j\{\A_i(\mathbf{1}_r)\} \geq \theta_{0,i} \eta_j\{\A_i(\tht_0)\} \geq c^{-1}\eta_j\{\A_i(\mathbf{1}_r)\}$, $j \in\{1,\dots, n\}$,
Chebyshev's inequality gives the representation
\begin{align} \label{approxofderiv}
\frac{\partial}{\partial\theta_i}  \big\|\C^{1/2}\uu_\dd \big\|^2
	&= O_P\left\{ \frac{\|\C^{1/2}\uu_\dd\|^2}{\theta_{0,i}}\right\}
	= O_P\left\{ \|\C^{1/2}\uu_\dd\|^2 \frac{\nu_i}{\sqrt{m}}\right\}, 
\end{align}
for all $\tht_0$. 
By Theorem~\ref{unifconv}, a Taylor expansion for $\|\hat\C^{1/2}\hat\uu_\dd \|^2$ around $\tht_0$ gives
\begin{align*}
\big\|\hat\C^{1/2}\hat\uu_\dd \big\|^2
	&= \big\|\C^{1/2}\uu_\dd \big\|^2 + O_P\bigg\{ \big(\hat\tht - \tht_0\big)^\top
		\frac{\partial}{\partial\tht}  \big\|\C^{1/2}\uu_\dd \big\|^2  \bigg\}
	= \big\|\C^{1/2}\uu_\dd \big\|^2 \big\{ 1 + O_P( m^{-1/2} ) \big\},
\end{align*}
which holds for all $\tht_0$.
Eventually, let $X_n = \|\C^{1/2}\uu_\dd\|^2/(1+\xi)$, $Y_n = O_P\{m^{-1/2}X_n + m^{-1/2}\xi/(1+\xi)\} = O_P\{b_nX_n\}$ for all $\tht_0$ with $b_n = m^{-1/2}$ and $a_n = t(1+\xi)^{-1}$.
Then,
\begin{align*}
\inf_{\tht_0\in\ThetaC} \Prob_{\tht_0}\left( \big\|\hat\C^{1/2}\hat\uu_\dd\big\|^2 \leq \hat{t} \right)
	&= \inf_{\tht_0\in\ThetaC} \Prob_{\tht_0}\left\{ X_n + Y_n \leq a_n \right\}.
\end{align*}
Finally, Lemma~\ref{pullerroutlemma} gives the claim.
\end{proof}
%
\begin{lemma} \label{minmization_wrt_d}
Let model (\ref{GenLM}) and (\ref{c_theta})-(\ref{c_linear}) hold. Then, for $t>0$,
\begin{align*}
\argmin_{\dd\in\{-1,1\}^p} \Prob_{\tht_0}\Big\{ \uu_\dd \in E\big(\C,t\big) \Big\}
	&= \argmax_{\dd\in\{-1,1\}^p} \big\| \C^{-1/2}\LAM\dd \big\|^2.
\end{align*}
\end{lemma}
This result is given in \cite[Prop. 4]{Ewald2018}.
%
\begin{proof}[\textbf{\upshape Proof of Theorem~\ref{mainres}:}]
Consider $\tau =   \max_{\dd} \chi^2_{p, 1-\alpha} ( \xi )$ with $\xi =  \| \C^{-1/2}\LAM\dd\|^2$.
%
Since for $X\sim \chi^2_p(\xi)$ it holds $X = O_P(1+\xi )$ for all $\tht_0$ and by the definition of the quantile $\Prob\big( X \leq \tau ) = 1 - \alpha$ it follows that $\tau = O(1 + \xi)$ for all $\tht_0$ as well. \\
%
Now we proceed similarly to the proof of Lemma~\ref{uniflimpiv}.
A Taylor expansion for $\|\hat\C^{-1/2}\LAM\dd \|^2$ around $\tht_0$ gives for $\OMEGA_i = \C^{-1}\X^\top\V^{-1}{\HH_i}\V^{-1}\X\C^{-1}/n$ that
\begin{align*}
\|\hat\C^{-1/2}\LAM\dd \|^2
	&= \|\C^{-1/2}\LAM\dd \|^2
		+ O_P\left\{  \sum_{i=1}^p \big( \hat\theta_i - \theta_{0,i} \big)
			\left\|   \OMEGA_i^{1/2} \LAM \dd \right\|^2  \right\} = \|\C^{-1/2}\LAM\dd \|^2
		+ O_P\big( m^{-1/2} \xi \big)
\end{align*}
for all $\tht_0$ and together with the first argument it follows that $\hat\tau = \tau + O_P(m^{-1/2}\xi)$ for all $\tht_0$.
By Lemma~\ref{uniflimpiv} it is ensured that the coverage is attained uniformly for both $\bta_0$ and $\tht_0$,
\begin{align*}
\inf_{\substack{\bta_0\in\mathbb{R}^p\\\tht_0\in\ThetaC}}\Prob_{\bta_0, \tht_0}\left\{ \sqrt{n}\left( \hat\bta_L - \bta_0 \right) \in E\left(\hat\C, \hat\tau  \right) \right\}
	&= \inf_{\substack{\dd\in\{-1,1\}^p\\\tht_0\in\ThetaC}}\Prob_{\tht_0}\left\{ \uu_\dd \in E\left(\C,\tau\right) \right\} + O(m^{-1/2}).
\end{align*}
By Lemma~\ref{minmization_wrt_d}, this minimum is in fact attained for the $\dd\in\{-1,1\}^p$ for which $\tau$ ensures nominal coverage, since $\| \C^{1/2} \uu_\dd \|^2 \sim \chi^2_{p}(\xi)$. This proves the claim.
\end{proof}

\section*{Acknowledgements}

This contribution has been partially created while the authors Peter
Kramlinger and Tatyana Krivobokova were at the University of
G\"ottingen. The authors gratefully acknowledge the funding by the
German Research Association (DFG) via Research Training Group 1644
``Scaling Problems in Statistics'' and thank anonymous referees for their insightful comments and suggestions.



\end{document}